\documentclass[acmsmall,10pt]{acmart}
\makeatletter\if@ACM@journal\makeatother
\acmJournal{PACMPL}
\startPage{1}
\else\makeatother
\acmYear{2018}
\startPage{1}
\fi

\setcopyright{none}

\usepackage{booktabs}   
\usepackage{subcaption} 

\setcitestyle{round}

\usepackage{listing}
\usepackage{amsthm,mathtools,amssymb,proof,stmaryrd}
\usepackage[ruled,algosection]{algorithm2e}
\usepackage{isomath}
\usepackage{multirow}
\usepackage{textcomp}
\usepackage{amscd}
\usepackage{amsfonts}
\usepackage{varwidth}
\usepackage{graphicx}
\usepackage{listings}
\usepackage{float}
\usepackage{multirow}
\usepackage[noend]{algorithmic}
\usepackage{mathrsfs}
\usepackage{mathpartir}
\usepackage{stmaryrd}
\usepackage{url}
\usepackage{hyperref}
\usepackage{xspace}
\usepackage{verbatim}
\usepackage{lipsum}
\usepackage{wrapfig}

\usepackage{enumitem}
\usepackage{url}
\usepackage{fancyvrb}
\usepackage[figurename=Figure]{caption}
\usepackage{textcomp}
\usepackage{nanoc}
\usepackage{ebproof}

\usepackage{multirow}
\usepackage{hhline}

\hypersetup{colorlinks,
  linkcolor=ACMDarkBlue,
  citecolor=ACMPurple,
  urlcolor=ACMDarkBlue,
  filecolor=ACMDarkBlue}

\SetAlFnt{\small}
\SetAlCapFnt{\small}
\SetAlCapNameFnt{\small}
\SetAlCapHSkip{0pt}
\IncMargin{-\parindent}

\floatstyle{plaintop}
\restylefloat{table}

\makeatletter 
\def\arcr{\@arraycr}
\makeatother

\definecolor{shadecolor}{gray}{1.00}
\definecolor{ddarkgray}{gray}{0.5}
\definecolor{darkgray}{gray}{0.30}
\definecolor{light-gray}{gray}{0.87}
\newcommand{\whitebox}[1]{\colorbox{white}{#1}}
\newcommand{\graybox}[1]{\colorbox{light-gray}{#1}}

\newcommand{\gbm}[1]{\graybox{${#1}$}}
\newcommand{\wbm}[1]{\whitebox{${#1}$}}

\newcommand{\naive}{na\"{i}ve\xspace}

\newcommand{\angled}[1]{\left\langle{#1}\right\rangle}
\newcommand{\paren}[1]{\left({#1}\right)}

\newcommand{\etc}{\emph{etc}\xspace}
\newcommand{\ie}{\emph{i.e.}\xspace}

\newcommand{\eg}{\emph{e.g.}\xspace}
\newcommand{\Eg}{\emph{E.g.}\xspace}

\newcommand{\dom}[1]{\mathsf{dom}\paren{#1}}

\newcommand{\cf}{\textit{cf.}\xspace}
\newcommand{\wrt}{\emph{wrt.}\xspace}
\newcommand{\Iff}{\emph{iff}\xspace}

\newcommand{\eqdef}{\triangleq}

\newcommand{\cached}[1]{#1}

\newcommand{\denot}[1]{\llbracket{#1}\rrbracket}

\newcommand{\ifext}[2]{\ifdefined\extflag{#1}\else{#2}\fi}

{ 
\theoremstyle{plain} 
{\unskip\nobreak\hskip 1em plus 1fil\nobreak$\square$
\parfillskip=0pt%
\endtrivlist}
}

\newcommand{\tname}[1]{\textsc{#1}\xspace}

\newcommand{\tool}{\tname{SuSLik}}
\newcommand{\logicfull}{Synthetic Separation Logic\xspace}
\newcommand{\logic}{SSL\xspace}
\newcommand{\exCount}{22\xspace}
\newcommand{\timeout}{120\xspace}

\definecolor{pblue}{rgb}{0.13,0.13,1}
\definecolor{pgreen}{rgb}{0,0.5,0}
\definecolor{pred}{rgb}{0.9,0,0}
\definecolor{pgrey}{rgb}{0.46,0.45,0.48}

\definecolor{ckeyword}{HTML}{7F0055}
\definecolor{ccomment}{HTML}{3F7F5F}
\definecolor{cnumber}{HTML}{2A0099}

\lstdefinelanguage{SynLang}{
  keywords={new, let, if, else, null, return, while},
  ndkeywords={bool, int, void, loc},
  mathescape=true,
  showspaces=false,
  showtabs=false,
  breaklines=true,
  showstringspaces=false,
  breakatwhitespace=true,
  lineskip=-0.9pt,
  morecomment=[l]{//}, 
  morecomment=[s]{/*}{*/}, 
  basewidth={0.54em, 0.4em},%
  basicstyle=\footnotesize\ttfamily,
  keywordstyle={\color{ACMPurple}\ttfamily\bfseries},
  ndkeywordstyle={\color{pblue}\ttfamily\bfseries},
  commentstyle={\color{ccomment}\itshape},
  numbers=none,
}

\lstdefinestyle{numbers}
{
  numbers=left,
  numberstyle=\scriptsize\sf,
  xleftmargin=15pt
}

\newcommand{\code}[1]{\lstinline[language=SynLang,basicstyle=\small\ttfamily,mathescape=true]{#1}}

\newcommand{\tth}{^{\text{th}}\xspace}
\newcommand{\pow}[1]{\wp({#1})}
\newcommand{\set}[1]{\left\{{#1}\right\}}
\newcommand{\many}[1]{\overline{#1}}

\newcommand{\head}[1]{\emph{#1}}

\makeatletter
\protected\def\ccell#1#{%
  \kern-\fboxsep
  \@ccell{#1}%
}
\def\@ccell#1#2#3{%
  \colorbox#1{#2}{#3}%
  \kern-\fboxsep
}
\makeatother

\newcommand{\pts}{\mapsto}
\newcommand{\True}{\mathsf{true}}

\newcommand{\bl}[1]{{\color{ACMDarkBlue}{#1}}}
\newcommand{\asn}[1]{{\bl{\set{#1}}}}
\newcommand{\spec}[3]{\asn{#1}~{#2}~\asn{#3}}
\newcommand{\pred}[1]{\mathsf{#1}}
\newcommand{\entails}[2]{{#1}\vdash{#2}}
\newcommand{\trans}[3]{\left.\bl{#1} \!\leadsto\! \bl{#2} \right| #3}
\newcommand{\tent}[2]{{#1}\leadsto {#2}}
\newcommand{\teng}[2]{\bl{#1}\leadsto \bl{#2}}

\newcommand{\Frameable}{\mathsf{frameable}}

\newcommand{\lev}{\ell}
\newcommand{\eqTags}{=^{\lev}}

\newcommand{\subst}{\sigma}

\newcommand{\dotcup}{\ensuremath{\mathaccent\cdot\cup}}
\newcommand{\union}{\mathbin{\dotcup}}
\newcommand{\seteq}{=}

\newcommand{\slf}[1]{\mathcal{#1}}
\newcommand{\slP}{\slf{P}}
\newcommand{\slQ}{\slf{Q}}
\newcommand{\interp}{\slf{I}}
\newcommand{\later}{\rhd}

\newcommand{\lseg}{\pred{lseg}}
\newcommand{\lsegt}{\pred{lseg2}}

\newcommand{\nxt}{\mathit{nxt}}
\newcommand{\sep}{\ast}
\newcommand{\Sep}{\text{\scalebox{1.8}{$\sep$}}}
\newcommand{\emp}{\mathsf{emp}}
\newcommand{\block}[2]{\left[{#1}, {#2}\right]}

\newcommand{\rulename}[1]{\textsc{#1}}
\newcommand{\framer}{\rulename{Frame}\xspace}
\newcommand{\hunir}{\rulename{UnifyHeaps}\xspace}

\newcommand{\hyp}{\rulename{UnifyPure}\xspace}
\newcommand{\starpar}{\rulename{StarPartial}\xspace}
\newcommand{\incons}{\rulename{Inconsistency}\xspace}
\newcommand{\empr}{\rulename{Emp}\xspace}
\newcommand{\nilvalr}{\rulename{NullNotLVal}\xspace}
\newcommand{\postinr}{\rulename{PostInconsistent}\xspace}
\newcommand{\postinl}{\rulename{PostInvalid}\xspace}
\newcommand{\readr}{\rulename{Read}\xspace}
\newcommand{\writer}{\rulename{Write}\xspace}

\newcommand{\sleft}{\rulename{SubstLeft}\xspace}
\newcommand{\sright}{\rulename{SubstRight}\xspace}
\newcommand{\heapunrf}{\rulename{UnreachHeap}\xspace}
\newcommand{\pickr}{\rulename{Pick}\xspace}

\newcommand{\allocr}{\rulename{Alloc}\xspace}
\newcommand{\freer}{\rulename{Free}\xspace}
\newcommand{\applyr}{\rulename{Call}\xspace}
\newcommand{\abductr}{\rulename{AbduceCall}\xspace}
\newcommand{\closer}{\rulename{Close}\xspace}
\newcommand{\invkind}{\rulename{Open}\xspace}
\newcommand{\mkind}{\rulename{Induction}\xspace}
\newcommand{\branchr}{\rulename{Branch}\xspace}

\newcommand{\ev}[1]{\mathsf{EV}\left({#1}\right)}

\newcommand{\vars}[1]{\mathsf{Vars}\left({#1}\right)}
\newcommand{\gv}[1]{\mathsf{GV}\left({#1}\right)}

\newcommand{\env}{\Gamma}
\newcommand{\ctx}{\Sigma}
\newcommand{\fctx}{\Delta}
\newcommand{\indp}{\mathcal{D}}
\newcommand{\fun}{\mathcal{F}}
\newcommand{\funs}{\mathsf{functions}}

\newcommand{\pclause}[3]{\angled{{#1}, \set{{#2}, {#3}}}}
\newcommand{\sel}{\xi}
\newcommand{\depth}{\mathsf{MaxUnfold}}
\newcommand{\bump}[2]{\lceil{#1}\rceil^{#2}}

\newcommand{\kw}[1]{{\tt\bf{\color{ACMPurple}{#1}}}~}
\newcommand{\typ}[1]{{\tt\bf{\color{blue}{#1}}}~}

\newcommand{\sym}[1]{{\tt\bf{{#1}}}}
\newcommand{\Letz}{\kw{let}}
\newcommand{\Ifz}{\kw{if}}
\newcommand{\malloc}[1]{{\tt{malloc}}(#1)}
\newcommand{\free}[1]{{\tt{free}}(#1)}

\newcommand{\Elsez}{\kw{else}}
\newcommand{\deref}[1]{*{#1}}

\newcommand{\off}{\iota}
\newcommand{\ccd}[1]{\text{\code{#1}}}

\newcommand{\coh}[1]{\mathsf{coh}\paren{#1}}

\newcommand{\prog}{c}
\newcommand{\progs}{\mathit{cs}}
\newcommand{\cskip}{{\tt{skip}}}
\newcommand{\cerror}{{\tt{error}}}
\newcommand{\cmagic}{{\tt{magic}}}
\newcommand{\opsem}[3]{{#1}; {#2} \rightsquigarrow {#3}}
\newcommand{\Opsem}{\rightsquigarrow}
\newcommand{\Opsemt}{\rightsquigarrow^{*}}
\newcommand{\opsemt}[3]{{#1}; {#2} \rightsquigarrow^{*} {#3}}
\newcommand{\satsl}[2]{{#1} \vDash^{\ctx}_{\interp} {#2}}

\SetAlgorithmName{Algorithm}{Algorithm}{List of Algorithms}
\DontPrintSemicolon
\SetKwInput{Input}{{\bf {\em Input\/}}}
\SetKwInput{Result}{{\bf {\em Result\/}}}

\SetKwProg{Fn}{function}{ = }{}
\SetKwSwitch{Switch}{Case}{Other}{match}{}{case}{default}{}{}
\SetKwFor{For}{for}{do}{} 
\newcommand{\tarr}{$~\Rightarrow$}
\SetInd{-0.2em}{0.8em}
\SetAlFnt{\footnotesize}
\LinesNumbered
\SetNlSty{textsf}{}{}
\newcommand{\goal}{\mathcal{G}}
\newcommand{\allrules}{\mathit{AllRules}}
\newcommand{\goals}{\mathsf{Goal}}
\newcommand{\goalz}{\mathit{goals}}
\newcommand{\deriv}{\mathcal{S}}
\newcommand{\derivs}{\mathsf{Deriv}}

\newcommand{\Rule}{\mathcal{R}}
\newcommand{\Rules}{\mathsf{Rule}}
\newcommand{\kont}{\mathcal{K}}
\newcommand{\konts}{\mathsf{Cont}}
\newcommand{\partialfn}{\rightharpoonup}

\newcommand{\rulez}{\mathit{rs}}
\newcommand{\trulez}{\mathit{rules}}
\newcommand{\Fail}{\mathsf{Fail}}
\newcommand{\sders}{\mathit{subderivs}}
\newcommand{\ders}{\mathit{derivs}}

\newcommand{\tryRules}{\mathsf{withRules}}
\newcommand{\tryAlts}{\mathsf{tryAlts}}
\newcommand{\solve}{\mathsf{solveSubgoals}}
\newcommand{\synt}{\mathsf{synthesize}}

\newcommand{\isInvert}{\mathit{isInvert}}
\newcommand{\nextRules}{\mathit{nextRules}}
\newcommand{\filterDervs}{\mathit{filterComm}}
\newcommand{\isEmpty}{\mathit{isEmpty}}

\newcommand{\pickRules}{\mathit{pickRules}}

\setlist[itemize]{leftmargin=*}
\setlist[enumerate]{leftmargin=*}

\begin{document}


\title[Structuring the Synthesis of Heap-Manipulating
Programs]{Structuring the Synthesis of Heap-Manipulating Programs}
\ifext{\subtitle{Extended Version}}{}

\author{Nadia Polikarpova}
\affiliation{%
  \institution{University of California, San Diego}
  \country{USA}
}

\author{Ilya Sergey}
\affiliation{%
  \institution{University College London}
  \country{UK}
}
 
\begin{abstract}
This paper describes a deductive approach to synthesizing imperative
programs with pointers from declarative specifications expressed in
Separation Logic.
Our synthesis algorithm takes as input a pair of assertions---a pre- and a postcondition---%
which describe two states of the symbolic heap,
and derives a program that transforms one state into the other,
guided by the shape of the heap.
%
%
The program synthesis algorithm rests on the novel framework of
\logicfull (\logic), which generalises the classical notion of heap
entailment $\entails{\slP}{\slQ}$ 
to incorporate a
possibility of transforming a heap satisfying an assertion $\slP$ 
into a heap satisfying an assertion $\slQ$.
%
A synthesized program represents a proof term for a \emph{transforming
  entailment} statement $\teng{\slP}{\slQ}$, and the synthesis
  procedure corresponds to a proof search.
The derived programs are, thus, correct by construction,
in the sense that they satisfy the ascribed pre/postconditions, and
are accompanied by complete proof derivations, which can be checked
independently.

We have implemented a proof search engine for \logic in a form the
program synthesizer called \tool. For efficiency, the engine exploits
properties of \logic rules, 
such as invertibility and commutativity of rule applications on separate heaps,
to prune the space of derivations it has to consider.
We explain and showcase the use of \logic on characteristic examples,
describe the design of \tool, and report on our experience of using it
to synthesize a series of benchmark programs manipulating
heap-based linked data structures.

\end{abstract}


\maketitle

\section{Introduction}
\label{sec:intro}

Consider the task of implementing a procedure \code{swap(x, y)}, 
which swaps the values stored in two distinct heap locations, 
\code{x} and \code{y}.
%
The desired effect of \code{swap} can be concisely captured via
pre/postconditions expressed in Separation Logic (SL)---a Hoare-style
program logic for specifying and verifying stateful programs with
pointers~\cite{Reynolds:LICS02,OHearn-al:CSL01}:
\begin{equation}
  \label{eq:swap}
  \asn{x \pts a \sep y \pts b}~\text{\code{void swap(loc x, loc y)}}~\asn{x \pts b \sep y \pts a}
\end{equation}


This specification is \emph{declarative}: 
it describes \emph{what} the heap should look like before and after 
executing \code{swap} without saying \emph{how} to get from one to
the other. 
Specifically, it
states that the program takes as input two pointers, \code{x} and
\code{y}, and runs in a heap where \code{x} points
to an unspecified value \code{a}, and \code{y} points to \code{b}.
Both \code{a} and \code{b} here are \emph{logical} (\emph{ghost})
variables, whose scope captures both pre- and
postcondition~\cite{Kleymann:FAC99}.
%
Because these variables are ghosts, 
we cannot use them directly to update the values in \code{x} and \code{y}
as prescribed by the postcondition;
the program must first ``materialize'' them
by readings them into local variables, \code{a2} and \code{b2} 
(\cf lines 2--3 of the code on the right).
\begin{wrapfigure}[7]{r}{0.37\textwidth}
\vspace{-15pt}
\begin{center}
\begin{lstlisting}[basicstyle=\small\ttfamily,style=numbers]
void swap(loc x, loc y) {
  let a2 = *x;
  let b2 = *y;
  *y = a2;
  *x = b2;
}
\end{lstlisting}
\end{center}
\end{wrapfigure}
%
In our minimalistic C-like language,
\texttt{\color{pblue}{\textbf{loc}}} denotes untyped pointers,
and \texttt{\textbf{\color{ckeyword}{let}}} introduces a local dynamically-typed variable.
Unlike in C, both formals and locals are \emph{immutable}
(the only mutation is allowed on the heap).

As a result of the two reads, the ghost variables in the post-condition can now be
substituted with equal program-level variables:
$\bl{x \pts \ccd{b2} \sep y \pts \ccd{a2}}$. 
This updated postcondition can be realized by the two writes on lines 4--5, 
which conclude our implementation, so the whole program can now be 
verified against the specification~\eqref{eq:swap}.
%

Having done this exercise in program derivation, let us now
observe that the SL specification has been giving us guidance on
what effectful commands (\eg, reads and writes) should be emitted
next.
%
In other words, the \emph{synthesis} of \code{swap} has been governed
by the given specification in the same way the \emph{proof search} is guided
by a goal in ordinary logics.
In this work, we make this connection explicit and employ it for
efficiently synthesizing imperative programs from SL pre- and
postconditions.

\paragraph{Motivation}
The goal of this work is to advance the
state of the art in synthesizing provably correct heap-manipulating programs 
from \emph{declarative functional} specifications. 
%
%
Fully-automated program synthesis has been an active area of research
in the past years, but recent techniques mostly targeted simple
DSLs~\cite{Gulwani-al:PLDI11,Polozov-Gulwani:OOPSLA15,Le-Gulwani:PLDI14}
or purely functional
languages~\cite{Polikarpova-al:PLDI16,Kneuss-al:OOPSLA13,Osera-Zdancewic:PLDI15,Feser-al:PLDI15}.
The primary reason is that those computational models impose strong
\emph{structural constraints} on the space of programs, either by means
of restricted syntax or through a strong type system. These structural
constraints enable the synthesizer to discard many candidate terms
\emph{a-priori}, before constructing the whole program, leading to
efficient synthesis.

Low-level heap-manipulating programs in general-purpose languages like
C or Rust lack inherent structural constraints \wrt control- and
data-flow, and as a result the research in synthesizing such programs
has been limited to cases when such constraints can be imposed by the
programmer.
From the few existing approaches we are aware of,
\tname{Simpl}~\cite{So-Oh:SAS17} and
\tname{ImpSynth}~\cite{Qiu-SolarLezama:OOPSLA17} require the
programmer to provide rather substantial \emph{sketches} of the control-flow
structure, which help restrict the search space; \tname{Jennisys}
by~\citet{Leino-Milicevic:OOPSLA2012} can only handle functions that
construct and read from data structures, but do not allow for
destructive heap updates, which are necessary for, \eg, deallocating,
modifying, or copying a linked data structure.

%


\paragraph{Key Ideas}
Our theoretical insight is that the structural constraints missing
from an imperative language itself, 
can be recovered from the \emph{program logic} 
used to reason about programs in that language.
We observe that synthesis of heap-manipulating programs 
can be formulated as a proof search in a generalized proof system 
that combines entailment with Hoare-style reasoning for 
\emph{unknown programs}.
In this generalized proof system,
a statement $\teng{\slP}{\slQ}$ means that
there \emph{exists} a program $\prog$, such that the Hoare triple
$\spec{\slP}{\prog}{\slQ}$ holds; 
the \emph{witness} program $\prog$ serves as a proof term for the statement.
%
In order to be useful, the system must satisfy a number of practical 
restrictions. 
First, it should be expressive enough to (automatically) verify the
programs with non-trivial heap manipulation. 
Second, it should be restrictive enough to make the synthesis problem
tractable.
Finally, it must ensure the termination of the (possibly recursive)
synthesized programs, to avoid vacuous proofs of partial correctness.

In this paper we design such a generalized proof system based on the
symbolic heap fragment of \emph{malloc/free} Separation
Logic\footnote{This nomenclature is due to \citet{Cao-al:APLAS17},
  who provide it as a rigorous alternative to the folklore notion of
  \emph{classical}~SL.} with inductive predicates, 
to which we will further refer as just
Separation Logic or SL~\cite{Reynolds:LICS02,OHearn-al:TOPLAS09}.
%
%
Separation Logic has been immensely successful at specifying and
verifying many kinds of heap-manipulating programs, both interactively and
automatically~\cite{Berdine-al:CAV11,Chlipala:PLDI11,Nanevski-al:POPL10,Appel-al:BOOK14,Chen-al:SOSP15,Chargueraud:ICFP10,Piskac-al:CAV14,Chin-al:SCP12,Distefano-Parkinson:OOPSLA08},
and is employed in modern symbolic execution 
tools~\cite{Rowe-Brotherston:CPP17,Berdine-al:APLAS05}. 
We demonstrate how to harness all this power for program synthesis,
devise the corresponding search procedure and apply it to synthesize a
number of non-trivial programs 
that manipulate linked data structures.
%
%
%
Finally, we show how to exploit laws of SL and properties of our proof system
to prune the search space and
make the synthesis machinery efficient for realistic examples.

\paragraph{Contributions}
\label{sec:contributions}
The central theoretical contribution of the paper is 
\emph{\logicfull} (\logic): a system of deductive synthesis rules,
which prescribe how to decompose specifications for complex programs
into specifications for simpler programs, while synthesizing the
corresponding computations compositionally.
In essence, \logic is a proof system for a new \emph{transformation}
judgment $\trans{\bl{\slP}}{\bl{\slQ}}{\prog}$ (reads as ``the assertion $\slP$
transforms into $\slQ$ via a program~$\prog$''), which unifies
SL entailment $\entails{\slP}{\slQ}$ and verification
$\spec{\slP}{\prog}{\slQ}$, with the former expressible as
$\trans{\bl{\slP}}{\bl{\slQ}}{\text{\code{skip}}}$.
%

The central practical contribution is the design and implementation of
\tool---a deductive synthesizer for heap-manipulating programs, based
on \logic.
%
%
%
\tool takes as its input a library of inductive predicates, 
a (typically empty) list of auxiliary function specifications, 
and an SL specification of the function to be synthesized. 
It returns a---possibly recursive, but loop-free---program (in a minimalistic C-like language), 
which \emph{provably} satisfies the given specification.

Our evaluation shows that \tool can synthesize all structurally-recursive
benchmarks from previous work on heap-based synthesis~\cite{Qiu-SolarLezama:OOPSLA17},
\emph{without any sketches} and in most cases much faster.
%
%
To the best of our knowledge, it is also \emph{the first synthesizer}
to automatically discover the implementations of copying linked lists
and trees, and flattening a tree to a list.

The essence of \tool's synthesis algorithm is a backtracking search in
the space of SSL derivations.
%
%
Even though the structural constraints (\ie, the shape of the heap)
embodied in the synthesis rules already prune the search space
significantly (as shown by our \code{swap} example), a \naive
backtracking search is still impractical, especially in the presence
of inductive heap predicates.
To eliminate redundant backtracking, we develop several principled
optimizations. In particular, we draw inspiration from \emph{focusing
  proof search}~\cite{Pfenning:LN04} to identify \emph{invertible}
synthesis rules that do not require backtracking, and exploit the
\emph{frame rule} of SL, observing that the order of rule applications
is irrelevant whenever their subderivations have disjoint footprints.


\paragraph{Paper outline}
In the remainder of the paper we give an overview of the reasoning
principles of \logic, describe its rules and the meta-theory, outline
the design and implementation of our synthesis tool, present the
optimizations and extensions of the basic search algorithm, and
report on the evaluation of the approach on a set of case studies
involving various linked structures, concluding with a discussion of
limitations and a comparison to the related work.

\section{Deductive Synthesis from Separation Logic Specifications}
\label{sec:overview}

In Separation Logic, assertions capture the program state, represented
by a symbolic heap.
An SL assertion (ranged over by symbols $\slP$ and $\slQ$ in the
remainder of the paper) is customarily represented as a pair
$\set{\phi; P}$ of a \emph{pure} part $\phi$ and a \emph{spatial}
part~$P$.
The \emph{pure} part (ranged over by $\phi$, $\psi$, $\xi$, and
$\chi$) is a quantifier-free boolean formula, which describes the
constraints over symbolic values (represented by variables $x$, $y$,
\etc)
%
The \emph{spatial} part (denoted $P$, $Q$, and $R$) is represented by
a collection of primitive heap assertions describing disjoint symbolic
heaps (\emph{heaplets}), conjoined by the \emph{separating
  conjunction} operation~$\sep$, which is commutative and
associative~\cite{Reynolds:LICS02}.
For example, in the assertion
$\asn{a \neq b ; x \pts a \sep y \pts b}$
the spatial part describes two disjoint memory cells
that store symbolic values $a$ and $b$,
while the pure part states that these values are distinct.

Our development is agnostic to the exact logic of pure formulae, as
long as it is decidable and supports standard Boolean connectives and
equality.\footnote{We require decidability for making the synthesis
  problem tractable, but it is not required for soundness of the
  logic.}
Our implementation uses the quantifier-free logic of arrays,
uninterpreted functions, and linear integer arithmetic,
which is efficiently decidable by SMT solvers,
and sufficient to express all examples in this paper.

%

To begin with our demonstration, the only kinds of heaplets we are
going to consider are the \emph{empty heap} assertion $\emp$ and
\emph{points-to} assertions of the form $\angled{x, \off} \pts e$,
where $x$ is a symbolic variable or pointer constant (\eg, 0), $\off$
is a non-negative integer offset (necessary to represent records and
arrays)
and $e$ is a symbolic value, stored in a memory cell, addressed via a
value of $(x + \off)$.\footnote{Further in the paper, we will extend
  the language of heap assertions to support memory blocks (arrays)
  with explicit memory management, as well as user-defined inductive
  predicates. }
In most cases, the offset is $0$, so we will abbreviate heap
assertions $\angled{x, 0} \pts e$ as $x \pts e$.

Our programming component (to be presented formally in
\autoref{sec:logic}) is a simple imperative language, supporting
reading from pointer variables to (immutable) local variables
(\code{let x = *y}), storing values into pointers (\code{*y = x}),
conditionals, recursive calls, and pure expressions.
The language has no \code{return} statement;
instead, a function stores its result into an explicitly passed pointer.

\subsection{Specifications for Synthesis}
A synthesis \emph{goal} is a triple
$\env; \tent{\bl{\slP}}{\bl{\slQ}}$, where $\env$ is an
\emph{environment}, \ie, a set of 
immutable program variables, $\slP$ is a \emph{precondition} (pre), and
$\slQ$ is a \emph{postcondition} (post).
Solving a synthesis goal means to find a program $c$ and a derivation
of the \logic assertion $\env; \trans{\bl{\slP}}{\bl{\slQ}}{\prog}$.
To avoid clutter, we employ the following naming conventions: 
%
\begin{itemize}[leftmargin=1.5em]

\item[(a)] the symbols $\slP$, $\phi$, and $P$ 
  refer to the goal's \emph{precondition}, its pure, and spatial part;

\item[(b)] similarly, the symbols $\slQ$, $\psi$, and $Q$ 
  refer to the goal's \emph{postcondition}, its pure and spatial part; 

\item[(c)] whenever the pure part of a SL assertion is $\True$
  ($\top$), it is omitted from the presentation.
\end{itemize}
\noindent
In addition to those conventions, we will use the following macros to
express the \emph{scope} and the \emph{quantification} over variables
of a goal ${\env}; \tent{\asn{\slP}}{\asn{\slQ}}$.
First, by $\vars{A}$ we will denote \emph{all} variables occurring in
$A$, which might be an assertion, a logical formula, or a program.
\emph{Ghosts} (universally-quantified logical variables), whose scope
is both the pre and the post, are defined as $\gv{\env, \slP, \slQ} =
\vars{\slP} \setminus \env$.
\emph{Goal existentials} are defined as
$\ev{\env, \slP, \slQ} = \vars{\slQ} \setminus \paren{\env \cup
  \vars{\slP}}$.
For instance, taking $\env = \set{x}$,
$\slP = \asn{x \neq y ; x \pts y}$, $\slQ = \asn{x \pts z}$, we have
the ghosts $\gv{\env, \slP, \slQ} = \set{y}$, the existentials
$\ev{\env, \slP, \slQ} = \set{z}$, and the pure part of the post
$\slQ$ is implicitly $\True$.

\subsection{Basic Inference Rules}
\label{sec:basic}
To get an intuition on how to represent program synthesis 
as a proof derivation in \logic, consider
\autoref{fig:basic-rules}, which shows four basic rules of the
logic, targeted to synthesize programs 
with constant memory footprint (remember that we use $\slP$ and $\slQ$
for the entire pre/post in a rule's conclusion!).

The \empr rule is applicable when both pre and post's
spatial parts are empty. 
It requires that no existentials remains in the goal, 
and the pure pre implies the pure post
(per our assumptions on the logic, the validity of this implication is
decidable, so we can check it algorithmically).
\empr has no synthesis subgoals among the premises (making it a
\emph{terminal} rule),
and no computational effect: 
its witness program is simply \code{skip}.

\begin{figure}[t]
\setlength{\belowcaptionskip}{-11pt}
\centering
\begin{tabular}{cc}
\begin{minipage}{0.4\linewidth}
\centering
{\scriptsize{
\!\!\!
\begin{mathpar}
\inferrule[\empr]
%
{
\ev{\env, \slP, \slQ} = \emptyset
\\
\phi \Rightarrow \psi
}
{
\env; \trans{\asn{\phi; \emp}}{\asn{\psi; \emp}}{\cskip}  
}
\\
\inferrule[\readr]
%
%
{
a \in \gv{\env, \slP, \slQ}
\\
y \notin \vars{\env, \slP, \slQ}
\\
\env \cup \set{y}; \trans{[y/a]\asn{\phi; \angled{x, \off} \pts a \sep P}}{[y/a]\asn{\slQ}}{\prog} 
}
{
\env; \trans{\asn{\phi; \angled{x, \off} \pts a \sep P}}
         {\asn{\slQ}}{\Letz y = \deref{(x + \off)}; \prog}  
}
\\
\inferrule[\writer]
%
%
{
\vars{e} \subseteq \env
\\
\env; \trans{{\asn{\phi; \angled{x, \off} \pts e \sep P}}}
         {\asn{\psi; \angled{x, \off} \pts e \sep Q}} {\prog}
}
{
\left.
\begin{array}{r@{\ }c@{\ }l}
\env; & \asn{\phi; \angled{x, \off} \pts e' \sep P} & \leadsto
\arcr[2pt]
& \asn{\psi; \angled{x, \off} \pts e \sep Q}
\end{array}
\right| {\deref{(x + \off)} = e; \prog}  
}
\\
\inferrule[\framer]
{
\ev{\env, \slP, \slQ} \cap \vars{R} = \emptyset
\\
\env; \trans{\asn{\phi; P}}{\asn{\psi; Q}}{\prog}  
}
{
\env; \trans{\asn{\phi; P \sep R}}{\asn{\psi; Q \sep R}}{\prog}  
}
\end{mathpar}
}}
\caption{Simplified basic rules of \logic.}
\label{fig:basic-rules}
\end{minipage}
&
\begin{minipage}{0.5\linewidth}
\centering
\vspace{0pt}
{\tiny{
\begin{prooftree}
\Infer0
[\empr with {$\prog_7$ = \tt{skip}}]
{
\begin{array}{c}
\set{x, y, {\tt{a2}}, {\tt{b2}}}; \tent{\asn{\emp}}{\asn{\emp}}   
  \\[3pt]
\prog_6 = \prog_7
\end{array}
}
\Infer1
[
\framer
]
{
\begin{array}{c}
\set{x, y, {\tt{a2}}, {\tt{b2}}}; 
\trans{\asn{\gbm{y \pts {\tt{a2}}}}}{\asn{\gbm{y \pts \tt{a2}}}}{\prog_6}   
  \\[3pt]
\prog_5 = \deref{y} = {\tt{a2}}; \prog_6
\end{array}
}
\Infer1
[
\writer
]
{
\begin{array}{c}
\set{x, y, {\tt{a2}}, {\tt{b2}}}; \trans{\asn{\gbm{y \pts {\tt{b2}}}}}{\asn{\gbm{y \pts \tt{a2}}}}{\prog_5}   
  \\[3pt]
\prog_4 = \prog_5
\end{array}
}
\Infer1
[
\framer
]
{
\begin{array}{c}
\set{x, y, {\tt{a2}}, {\tt{b2}}}; 
\trans{\asn{\gbm{x \pts \tt{b2}} \sep y \pts
  {\tt{b2}}}}{\asn{\gbm{x \pts {\tt{b2}}} \sep y \pts \tt{a2}}}{\prog_4}   
  \\[3pt]
\prog_3 = \deref{x} = {\tt{b2}}; \prog_4
\end{array}
}
\Infer1
[
\writer
]
{
\begin{array}{c}
\set{x, y, {\tt{a2}}, {\tt{b2}}}; 
\trans{\asn{\gbm{x \pts \tt{a2}} \sep y \pts
  {\tt{b2}}}}{\asn{\gbm{x \pts {\tt{b2}}} \sep y \pts \tt{a2}}}{\prog_3}   
  \\[3pt]
\prog_2 = \Letz {\tt{b2}} = \deref{y}; \prog_3
\end{array}
}
\Infer1
[
\readr
]
{
\begin{array}{c}
\set{x, y, \tt{a2}}; 
  \trans{\asn{x \pts \tt{a2} \sep \gbm{y \pts  b}}}{\asn{x \pts b \sep y \pts \tt{a2}}}{\prog_2}
  \\[3pt]
  \prog_1 = \Letz {\tt{a2}} = \deref{x}; \prog_2
\end{array}
}
\Infer1
[
\readr
]
{
\set{x, y}; \trans{\asn{\gbm{x \pts a} \sep y \pts b}}{\asn{x \pts b \sep y \pts a}}{\prog_1}
}
\end{prooftree}
}}
\caption{Derivation of \code{swap(x,y)} as $c_1$.}
\label{fig:swap}
\end{minipage}  
\end{tabular}
\end{figure}

The \readr rule turns a ghost variable $a$ into a program
variable $y$ (fresh in the original goal). That is, the newly assigned
immutable program variable $y$ is added to the environment of the
sub-goal, and all occurrences of $a$ are substituted by $y$ in both
the pre and post. As a side-effect, the rule prepends the read
statement $\Letz~y = \deref(x + \off)$ to the remainder of the program
to be synthesized.

The rule \writer allows for writing a symbolic expression $e$ 
into a memory cell, provided 
all $e$'s variables are program-level. 
It is customary for
this rule to be followed by an application of {\framer}, which is
\logic's version of Separation Logic's \emph{frame rule}.
Here, we show a version of the rule, which is a bit weaker
than what's in the full version of \logic, and will be generalized
later.
This rule enables ``framing out'' a shared sub-heap $R$ from the pre and post,
as long as this does not create new existential variables.
Notice, that unlike the classical SL's \framer rule
by~\citet{OHearn-al:CSL01}, our version does not require a side
condition saying that $R$ must not contain program variables that are
modified by the program [to be synthesized]: by removing $R$ from the
subgoal, we ensure that the residual program will not be able to
access any pointers from $R$, because it will be synthesized in a
symbolic footprint \emph{disjoint} from~$R$, and all local variables in
\logic language are \emph{immutable}.

%
%
%


\paragraph{Synthesizing \texttt{swap}}
\label{sec:swap}

Armed with the basic \logic inference rules from
\autoref{fig:basic-rules}, let us revisit our initial example:
the \code{swap} function~\eqref{eq:swap}.
\autoref{fig:swap} shows the derivation of the program using the
rules, and should be read bottom-up. For convenience, we name each
subgoal's witness program, starting from $c_1$ (which corresponds to
\code{swap}'s body). 
Furthermore, each intermediate sub-goal highlights via
\graybox{gray boxes} a part of the pre and/or the post, which
``triggers'' the corresponding \logic rule.
%
%
Intuitively, the goal of the synthesis process is to ``empty'' the spatial parts 
of the pre and the post, so that the derivation can eventually be closed via \empr;
to this end, \readr and \writer work together to create matching heaplets
between the two assertions, which are then eliminated by \framer.


\subsection{Spatial Unification and Backtracking}

Now, consider the synthesis goal induced by the following
SL specification:

\vspace{-5pt}

\begin{equation}
\label{eq:pick}
\asn{ x \pts 239 \sep y \pts 30}
~\text{\code{void pick(loc x, loc y)}}~
\asn{ x \pts z \sep y \pts z}  
\end{equation}

\vspace{5pt}

Since $z$ does not appear among the formals or in the precondition, 
it is treated as an existential. 
The postcondition thus allows $x$ and $y$
to point to any value, as long as it is the same value.



\begin{wrapfigure}[8]{r}{0.4\textwidth}
\setlength{\abovecaptionskip}{0pt}
\vspace{-10pt}
\centering
\begin{minipage}{0.9\linewidth}
{\footnotesize{
\centering
\begin{mathpar}
\!\!\!\!\!\!\!\!\!\!\!\!
\inferrule[\hunir]
{
[\subst]R' = R
\\
\emptyset \neq \dom{\subst} \subseteq \ev{\env, \slP, \slQ}
\\
\env; \trans{\asn{P \sep R}}{[\subst]\asn{\psi; Q \sep R'}}{\prog}
}
{
\env; \trans{\asn{\phi; P \sep R}}{\asn{\psi; Q \sep R'}}{\prog}  
}
\end{mathpar}
}}
\end{minipage}
\caption{\logic rule for heap unification.}
\label{fig:frame-write}
\end{wrapfigure}
To deal with existentials in the heap, we introduce the rule \hunir,
which attempts to find a unifying substitution $\subst$ for some sub-heaps 
of the pre and the post.
The domain of $\subst$ must only contain existentials.
%
%
For example, applying \hunir to the spec~\eqref{eq:pick} with
$R \eqdef x \pts 239$ and $R' \eqdef x \pts z$ results in the
substitution $\subst = [z \mapsto 239]$, and the residual synthesis
goal
$\set{x, y} \teng{\asn{x \pts 239 \sep y \pts 30}}{\asn{x \pts 239
    \sep y \pts 239}}$, which can be now synthesized by using the
\framer ,\writer, and \empr rules.


Due to their freedom to choose a sub-heap (and a unifying substitution),
\framer and \hunir introduce non-determinism into the synthesis procedure
and might require backtracking---a fact also widely observed in
interactive verification
community~\cite{McCreight-TPHOL09,Gonthier-al:ICFP11} \wrt SL
assertions.
%
For instance, consider the spec below:
\begin{equation}
  \label{eq:notsure}
  \asn{x \pts a \sep y \pts b}
  ~\text{\code{void notSure(loc x, loc y)}}~
  \asn{x \pts c \sep c \pts 0}
\end{equation}
%

\begin{wrapfigure}[3]{r}{0.4\textwidth}
\vspace{-20pt}
\begin{center}
\begin{lstlisting}[basicstyle=\footnotesize\ttfamily]
void notSure(loc x, loc y) {
  *x = y;
  *y = 0;
}\end{lstlisting}
\end{center}
\end{wrapfigure}
One way to approach the spec~\eqref{eq:notsure} is to first read from $x$, 
making $a$ a program-level variable \code{a2} (via \readr),
then use \hunir and \framer on the $\bl{x \pts \bullet}$ heaplets in
the pre/post,
substituting the existential $c$ by~$a2$. 
That, however, leaves us with an unsolvable goal
$\set{x,y, \ccd{a2}} \tent{\asn{y \pts b}}{\asn{\ccd{a2} \pts 0}}$.
Hence we have to backtrack, and instead unify $c$ with $y$, eventually
deriving the correct program \code{notSure}.


\subsection{Reasoning with Pure Constraints}
\label{sec:purec}

So far we have only looked at SL specifications whose \emph{pure}
parts were trivially $\True$. Let us now turn our attention to the
goals that make use of non-trivial pure boolean assertions.

\subsubsection{Preconditions}
To leverage pure \emph{preconditions}, we adopt
a number of the traditional \tname{Smallfoot}-style rules, whose
\logic counterparts are shown in the top part of
\autoref{fig:pure}.
In the nomenclature of \citet{Berdine-al:APLAS05}, all those rules are
\emph{non-operational}, \ie, correspond to constructing the proofs of
symbolic heap entailment and involve no programming component.
Note that the original rules in \citet{Berdine-al:APLAS05} assume a 
restricted pure logic with only equalities;
we adapt these rules to our logic-agnostic style,
relying on the oracle for pure validity instead of original syntactic premises.   

For instance, the rule \sleft makes use of a precondition that implies 
equality between two universal variables, $x = y$,
substituting all occurrences of $x$ in the subgoal by $y$.
%
The rule \starpar makes explicit the fundamental assumption of SL: 
disjointness of symbolic heaps connected by the $\sep$ operator.
%
Most commonly, this rule's effect is observable in combination with
another rule, \incons, which identifies an inconsistent pre, 
and emits an always-failing program \code{error}. 

These three \logic rules can be observed in action via the following
example:
\begin{equation}
  \label{eq:urk}
{\small{
\asn{a = x  \wedge  y = a ; x \pts y \sep y \pts z}
  ~\text{\code{void urk(loc x, loc y)}}~
  \asn{\True; y \pts a \sep x \pts y}
}}
\end{equation}

\begin{figure}[t]
\setlength{\belowcaptionskip}{-10pt}
\centering
\begin{tabular}{c}
\begin{tabular}{ccc}
\begin{minipage}{0.25\linewidth}
{\footnotesize{
\centering
\begin{mathpar}
\!\!\!\!\!\!\!\!\!\!\!\!
\inferrule[\sleft]
{
\phi \Rightarrow x = y
\\
\env; \trans{[y/x]\asn{\phi; P}}{[y/x]\asn{\slQ}}{\prog}
}
{
\env; \trans{\asn{\phi; P}}{\asn{\slQ}}{\prog}
}
\end{mathpar}
%
}}
\end{minipage}
&
\begin{minipage}{0.42\linewidth}
{\footnotesize{
\centering
\begin{mathpar}
\!\!\!\!\!\!\!\!\!\!\!\!
\inferrule[\starpar]
{
x + \off \neq y + \off' \notin \phi
\\
\phi' = \phi \wedge (x + \off \neq y + \off')
\\
\env; 
\trans{\asn{\phi'; 
    \angled{x, \off} \pts e \sep \angled{y, \off'} \pts e' \sep P}}
{\asn{\slQ}}{\prog}  
}
{
\env; 
\trans{\asn{\phi; 
    \angled{x, \off} \pts e \sep \angled{y, \off'} \pts e' \sep P}}
{\asn{\slQ}}{\prog}  
}
\end{mathpar}
}}
\end{minipage}
&
\begin{minipage}{0.27\linewidth}
{\footnotesize{
\centering
\begin{mathpar}
\!\!\!\!\!\!\!\!\!\!\!\!
\inferrule[\incons]
{\phi \Rightarrow \bot}
{
\env; 
\trans{\asn{\phi ; P}}
{\asn{\slQ}}{\cerror}  
}
\end{mathpar}
}}
\end{minipage}
\end{tabular}
\\  
\begin{tabular}{ccc}
\begin{minipage}{0.33\linewidth}
{\footnotesize{
\centering
\begin{mathpar}
\!\!\!\!\!\!\!\!\!\!\!\!
\inferrule[\sright]
{
x \in \ev{\env, \slP, \slQ}
\\
\ctx; \env; \trans{\asn{\slP}}{[e/x]\asn{\psi, Q}}{\prog}
}
{
\ctx; \env; \trans{\asn{\slP}}{\asn{\psi \wedge x = e; Q}}{\prog}
}
\end{mathpar}
}}
\end{minipage}
&
\begin{minipage}{0.33\linewidth}
{\footnotesize{
\centering
\begin{mathpar}
\!\!\!\!\!\!\!\!\!\!\!\!\!\!\!\!\!\!\!\!\!\!\!\!\!
\inferrule[\pickr]
{
y \in \ev{\env, \slP, \slQ}
\\
\vars{e} \in \env \cup \gv{\env, \slP, \slQ}
\\
\env; \trans{\asn{\phi; P}}
         {[e/y]\asn{\psi; Q}}{\prog}  
}
{
\env; \trans{\asn{\phi; P}}
         {\asn{\psi; Q}}{\prog}  
}
\end{mathpar}
}}
\end{minipage}

&
\begin{minipage}{0.33\linewidth}
{\footnotesize{
\centering
\begin{mathpar}
\!\!\!\!\!\!\!\! \!\!\!\!\!\!\!\! \!\!\!\!\!\!\!\!
\inferrule[\hyp]
{
[\subst]\psi' = \phi'
\\
\emptyset \neq \dom{\subst} \subseteq \ev{\env, \slP, \slQ}
\\
\env; \trans{\asn{\slP}}{[\subst]\asn{\slQ}}{\prog}
}
{
\env; \trans{\asn{\phi \wedge \phi'; P}}{\asn{\psi \wedge \psi'; Q}}{\prog}  
}
\end{mathpar}
}}
\end{minipage}
\end{tabular}
\end{tabular}
\caption{Selected \logic rules for reasoning with pure constraints in
  the synthesis goal.}
\label{fig:pure}
\end{figure}

After applying \sleft, the goal transforms to
$\set{x, y} \tent{\asn{x \pts x \sep x \pts z}}{\asn{x \pts x \sep x
    \pts x}}$,
which is clearly unsatisfiable, as the precondition requires two
\emph{disjoint} points-to heaplets with the same source---a fact,
which converted into a pure sub-formula $x \neq x$ by
\starpar, resulting in the \code{error} body via
\incons.

\subsubsection{Postconditions}\label{sec:pure-synthesis}

In the presence of non-trivial pure \emph{postconditions}, we face the
problem to find suitable instantiations for their existentials. This
is a challenging, yet well-studied problem, tackled by \emph{pure
  program synthesis}~\cite{Alur-al:FMCAD13}. We consider this problem
orthogonal to our agenda of deriving pointer-manipulating programs,
and represent pure synthesis with a simplification rule
\sright~\cite{Berdine-al:APLAS05}, exploiting an equality in a goal's
postcondition, and an oracle rule \pickr, which picks an instantiation
for an existential non-deterministically,

In practice, the non-determinism can be curbed, for example, by
delegating to an existing pure
synthesizer~\cite{Kuncak-al:PLDI10,Reynolds-al:CAV15}. In our
implementation, however, we found a combination of first-order
unification (rule \hyp) and restricted enumerative search (rule \pickr
restricted to variables) to be very effective at discharging such
synthesis goals.


%
%
As an example, consider the following goal
(where $S$ and $S_1$ are finite sets and $\union$ is disjoint union):
%
\begin{equation}
  \label{eq:elem}
{\small{
\asn{S \seteq \set{\ccd{v}} \union S_1; \ccd{x} \pts a}
  ~\text{\code{void elem(loc x, int v)}}~
  \asn{S \seteq \set{v_1} \union S_1; \ccd{x} \pts v_1 + 1}
}}
\end{equation}
%
Following the rule \hyp, one can unify the two facts about sets in the
pre and the post, 
obtaining the substitution $[v_1 \mapsto v]$.
The rest is accomplished by the rule
\writer, which emits the only necessary statement for \code{elem}'s
body: \code{*x = v + 1}.

\subsection{Synthesis with Inductive Predicates}

The real power of Separation Logic stems from its ability to
compositionally reason about linked heap-based data structures, such
as lists and trees, whose shape is defined recursively via
\emph{inductive heap predicates}.
The most traditional example of a data structure defined this way is a
linked list segment $\lseg(x, y, S)$~\cite{Reynolds:LICS02}, whose
definition is given by the two-clause predicate below:

\vspace{-5pt}
\begin{equation}
\label{eq:lseg}
\begin{array}{r@{\ \ }c@{\ \ }r@{\ }c@{\ }l}
\lseg(x, y, S) & \eqdef & x = y & \wedge & \set{ S \seteq \emptyset ; \emp}
\\[3pt]
& & x \neq y & \wedge & \set{ S \seteq \set{v} \union S_1 ; 
\block{x}{2} \sep x \pts v \sep \angled{x, 1} \pts \nxt \sep \lseg(\nxt, y, S_1) } 
\end{array}
\end{equation}
%
\vspace{5pt}

The predicate definition, which we will abstractly denote as
$\indp \eqdef p(\many{x_i})\many{\pclause{\sel_j}{\chi_j}{R_j}}_{j \in
  1\ldots N}$,
starts from the name $p$ and a vector or formal parameters
$\many{x_i}$; for $\lseg$ those are the symbolic pointer variables,
$x$ and $y$ for the first and the last pointer in the list, as well as
for the logical set $S$ of its elements.
%
What follows is a sequence of $N$ \emph{inductive clauses}, with a
$j\tth$ clause starting from a \emph{guard} $\sel_j$---a
\emph{boolean} formula defining a condition on the predicate's
formals,\footnote{In the case of logical overlap, the conditions for
  different clauses are checked in the order the clauses are defined.}
followed by the \emph{clause body}---a SL assertion with a spatial
part $R_j$ and pure part $\chi_j$, describing the shape of the heap
and pure constraints, correspondingly.
Clauses' free variables, \ie, non-formals, (\eg, $v$, $\nxt$) are
treated as ghosts or existentials, depending on whether the predicate
instance is in a pre or post of a goal.
From now on, we extend the definition of the goal, with a
\emph{context} $\ctx$, which will store the definitions of inductive
predicates and specified functions, which are accessible in the
derivation.

That is, the first clause of $\lseg$ states that in the case of $x$
and $y$ being equal, the linked list's set of element is empty and its
implementation is an empty heap $\emp$.
The complementary second clause postulates the existence of the
\emph{allocated memory block} (or just \emph{block}) of two
consecutive pointers \emph{rooted} at $x$ (denoted
$\block{x}{2} \sep x \pts v \sep \angled{x, 1} \pts \nxt$), such that
the first pointer stores the payload $v$, while the second one points
to the tail of the list, whose shape is defined recursively via the
same predicate, although with different actuals, as captured by the
predicate instance $\lseg(\nxt, y, S_1)$.

\subsubsection{Dynamic Memory}

In order to support dynamically allocated linked structures, 
as demonstrated by
definition~\eqref{eq:lseg}, we extend the language of symbolic heaps
with two new kinds of assertions: \emph{blocks} and \emph{predicate
  instances}.
Symbolic blocks are a well-established way to add to SL support for
consecutive memory
chunks~\cite{Jacobs-al:NFM11,Brotherston-al:CADE17}, which can be
allocated and disposed all together.\footnote{An alternative would be to
  adopt an object model with \emph{fields}, which is more
  verbose~\cite{Berdine-al:APLAS05}.}
Two \logic rules, \allocr and \freer (presented in
\autoref{sec:logic}), make use of blocks, as those appear in the post-
and the pre-conditions of their corresponding goals.
%
%
Conceptually, 
%
\allocr looks for a block in the postcondition rooted at an existential 
and allocates a block of the same size
(by emitting the command \code{let x = malloc(n)}), 
adding it to the subgoal's precondition.
%
\freer is triggered by an un-matched block in the goal's pre,
rooted at some program variable \code{x}, which it then disposes by
emitting the call to \code{free(x)}, removing it from the subgoal's
precondition.
%

\subsubsection{Induction}
Let us now synthesize our first recursive heap-manipulating function,
a linked list's \emph{destructor} \code{listfree(x)}, which expects a
linked list starting from its argument $x$ and ending with the
null-pointer, and leaves an empty heap as its result:
%
%
\begin{equation}
\label{eq:list-free}
\asn{ \lseg(x, 0, S)}
~\text{\code{void listfree(loc x)}}~
\asn{\emp}    
\end{equation}
%

The first synthesis step is carried out by the \logic rule \mkind. 
We postpone its formal description until the
next section, conveying the basic intuition here.
\mkind only applies to the \emph{initial} synthesis goal
whose precondition contains an inductive predicate instance,
and its effect is to add a new \emph{function symbol} to the goal's context,
such that an invocation of this function would correspond to a recursive call.
%
In our example~\eqref{eq:list-free} \mkind extends the context $\ctx$
with a ``recursive hypothesis'' as follows (we explain the meaning of
the tag 1 in $\bl{\lseg^1(x', 0, S')}$ later):\footnote{In \logic, a
  context $\ctx$ can also store user-provided specifications of
  auxiliary functions synthesized earlier. We will elaborate on case
  studies relying on user-provided auxiliary functions in
  \autoref{sec:results}.}

%
%
\begin{equation}
  \label{eq:frec}
  \ctx_1 \eqdef \ctx, \ccd{listfree}(x'): \asn{\lseg^1(x', 0, S')}\asn{\emp}
\end{equation}
%

\subsubsection{Unfolding Predicates}

The top-level rule \mkind is complemented by the rule \invkind 
(defined in \autoref{sec:logic}), which \emph{unfolds} a predicate instance
in the goal's precondition according to its definition,
and creates a subgoal for each inductive clause.
For instance, invoked immediately on our
goal~\eqref{eq:list-free}, it has the following effect on the
derivation:

\begin{itemize}[leftmargin=1.5em]
\item[(a)] Two sub-goals, one for each of the clauses of $\lseg$, are
  generated to solve:

  \begin{itemize}
  \item[($i$)]  $\ctx_1 ; \set{x} ; 
    \tent{\asn{x = 0 \wedge S \seteq \emptyset; \emp}} {\asn{\emp}}$

  \item[($ii$)] $\ctx_1 ; \set{x} ;
    \tent{\asn{x \neq 0 \wedge S \seteq \set{v} \union S_1; 
        \block{x}{2} \sep x \pts v \sep \angled{x, 1} \pts \nxt \sep \lseg^1(\nxt, y, S_1)}} {\asn{\emp}}$
  \end{itemize}

\item[(b)] Assuming $c_1$ and $c_2$ are the programs solving the
  sub-goals ($i$) and ($ii$), the final program is obtained by
  combining them as $\ccd{if}~(x~\sym{=}~0)~\set{c_1}~\ccd{else}~\set{c_2}$.
\end{itemize}
Thus, \invkind performs case-analysis according to the predicate definition. 
Note how the precondition of each generated sub-goal
is \emph{refined} by the corresponding clause's guard and body.
%
%
The resulting sub-programs, once synthesized, are then combined with
the conditional statement (this is why we require decidability of the
guard statements), which branches on the predicate's guard.
%
%
It is easy to see that the first sub-goal ($i$) can be immediately
solved via \empr rule, producing the program \code{skip}.

\subsubsection{Level Tags}

Synthesizing recursive programs requires extra care in order to avoid
infinite derivations, as well as vacuously correct (in the sense of
partial program correctness) \emph{non-terminating} programs that simply call themselves.
To avoid this pitfall, we adopt the ideas from the Cyclic Termination
Proofs in SL~\cite{Brotherston-al:APLAS12}, under the assumption that
employed user-defined inductive predicates are \emph{well-founded},
\ie, have their recursive applications only on strictly smaller
sub-heaps~\cite{Brotherston-al:POPL08}.
We ensure that this is the case by checking that there is at least one
points-to predicate in clauses that also contain predicate instances.

Specifically, to avoid infinite unfolding of predicate instances
we introduce
\emph{level tags} (natural numbers, ranged over by~$\lev$), which now
annotate some predicate instances in the pre and post
of the goal and the context functions.
%
For an instance in a goal's pre, a tag determines whether a set of
functions in $\ctx$ that can be ``applied'' to it.
As a result of a function application, tags are modified 
(as we explain later), 
thus preventing functions from being ``re-applied'' to their own symbolic post-heaps.
Since tags only serve to control function calls, 
the rules \framer (\autoref{fig:basic-rules}) and \hunir (\autoref{fig:frame-write}) ignore them 
when comparing sub-heaps for equality.
%
All predicates in the pre/post of the initial goal have their level
tag set as $\lev = 0$, and the rule \invkind only applies to
$0$-tagged predicates, incrementing their tag (\ie, one cannot
``unfold again'' an already opened instance).

%

\begin{figure}[t]
\setlength{\abovecaptionskip}{-0pt}
\setlength{\belowcaptionskip}{-10pt}
\centering
\begin{tabular}{cc}
\begin{minipage}{0.5\linewidth}
{\footnotesize{
\centering
\begin{mathpar}
\!\!\!\!\!\!\!\!\!
\inferrule[\applyr]
{
\fun \eqdef f(\many{x_i}) : \asn{\phi_f, P_f}\asn{\psi_f, Q_f} \in \ctx
\\
R \eqTags [\sigma]P_f
\\
\phi \Rightarrow [\sigma]\phi_f
\\
\phi' \eqdef [\subst]\psi_f
\quad\quad
R' \eqdef \bump{[\subst]Q_f}{}
\\
\many{e_i} = [\subst]\many{x_i}
\\
\vars{\many{e_i}} \subseteq \env
\\
\ctx; \env; \trans{\asn{\phi \wedge \phi'; P \sep R'}}{\asn{\slQ}}{\prog}  
}
{
\ctx; \env; \trans{\asn{\phi; P \sep R}}{\asn{\slQ}}{f(\many{e_i}); \prog}  
}
\end{mathpar}
}}
\end{minipage}
&
\!\!\!\!\!\!\!\!\!\!
\begin{minipage}{0.5\linewidth}
{\footnotesize{
\centering
\begin{mathpar}
\inferrule[\closer]
{
\indp \eqdef p(\many{x_i})\many{\pclause{\sel_j}{\chi_j}{R_j}}_{j \in
  1\ldots N} \in \ctx
\\
\lev < \depth
\\
1 \leq k \leq N
\\
\subst \eqdef [\many{x_i \mapsto y_i}]
\\
R' \eqdef \bump{[\subst] R_k}{\lev + 1}
\\
\ctx;\env; \trans{\asn{\slP}}
         {\asn{\psi \wedge  [\subst]\sel_k \wedge [\subst]\chi_k; Q \sep R'}}{\prog}  
}
{
\ctx;\env; \trans{\asn{\slP}}
         {\asn{\psi; Q \sep p^{\lev}(\many{y_i})}}{\prog}  
}
\end{mathpar}
}}
\end{minipage}
\end{tabular}
\caption{Selected \logic rules for synthesis with recursive functions
  and inductive predicates.}
\label{fig:induction}
\end{figure}

To see how tags control what functions from $\ctx$ can be applied,
consider the rule~\applyr in \autoref{fig:induction}.
It fires when the goal contains in its precondition a symbolic
sub-heap $R$, which can be unified with the precondition $P_f$ of a
function symbol $f$ from the goal's context $\ctx$.
This unification is similar to the effect of \hunir,
with the difference that \applyr takes level tags into the account 
(\ie, instances with different tags cannot be unified),
reflected in the tag-aware equality predicate $\eqTags$.
%
%
Our example's second goal~($ii$)
\begin{equation}
  \label{eq:list-free-1}
  \set{x} ; \tent{\asn{x \neq 0 \wedge S \seteq \set{v} \union S_1;
    \block{x}{2} \sep x \pts v \sep \gbm{\angled{x, 1} \pts \nxt} \sep
    \lseg^1(\nxt, y, S_1)}} {\asn{\emp}}
\end{equation}
%
%
can be now transformed, via \readr (focused on $\bl{\angled{x, 1} \pts
  \nxt}$), into
\begin{equation}
  \label{eq:list-free-2}
\set{x, \ccd{nxt2}} ; \tent{\asn{x \neq 0 \wedge S \seteq \set{v} \union
    S_1; \block{x}{2} \sep x \pts v \sep \angled{x, 1} \pts \ccd{nxt2} \sep
    \gbm{\lseg^1(\ccd{nxt2}, y, S_1)}}} {\asn{\emp}}
\end{equation}

\begin{wrapfigure}[8]{r}{0.38\textwidth}
\setlength{\abovecaptionskip}{-3pt}
\vspace{-17pt}
\begin{center}
\begin{lstlisting}[basicstyle=\small\ttfamily]
void listfree(loc x) {
  if (x == 0) {} else {
    let nxt2 = *(x + 1);
    listfree(nxt2);
    free(x);
  } }
\end{lstlisting}
\end{center}
\caption{Synthesized \code{listfree}~\eqref{eq:list-free}.}
\label{fig:free}
\end{wrapfigure}
The grayed fragment in~\eqref{eq:list-free-2} can now be unified with
the precondition of \code{listfree}~\eqref{eq:frec} following
\applyr's premise.
As the tags match 
(both indicate the \emph{first} unfolding of the predicate), 
unification succeeds with the substitution
$\subst = [x' \mapsto \ccd{nxt2}, S' \mapsto S_1]$ 
from $f$'s parameters and ghosts to the goal variables.
%
The same rule produces, 
from the $f$'postcondition,
a new symbolic heap $R'$, which replaces the targeted fragment in
the pre, recording the effect of the call.
All tagged predicate instances in $R'$ get their tags \emph{erased}
($\bump{[\subst]Q_f}{}$), thus, preventing any future recursive
applications (via \applyr) on the produced symbolic heap (more on that
design decision in \autoref{sec:discussion}).
%
However, in this example, the function's post is merely $\asn{\emp}$,
so the goal becomes:
\begin{equation}
  \label{eq:list-free-3}
\set{x, \ccd{nxt2}} ; \tent{\asn{x \neq 0 \wedge S \seteq \set{v} \union
    S_1; \block{x}{2} \sep x \pts v \sep \angled{x, 1} \pts \ccd{nxt2}
    \sep \emp}} {\asn{\emp}}
\end{equation}
%
%

\begin{wrapfigure}[12]{r}{0.40\textwidth}
\setlength{\abovecaptionskip}{0pt}
\vspace{-20pt}
\begin{center}
\begin{lstlisting}[basicstyle=\footnotesize\ttfamily,xleftmargin=4pt]
void listmorph(loc x, loc r) {
  if (x == 0) { } else {
    let v2 = *x;
    let nxt2 = *(x + 1);
    listmorph(nxt2, r);
    let y12 = *r;
    let y2 = malloc(3);
    free(x);
    *(y2 + 2) = y12;
    *(y2 + 1) = v2 + 1;
    *y2 = v2;
    *r = y2;
  } }
\end{lstlisting}
\end{center}
\caption{Synthesized \code{listmorph}~\eqref{eq:listmorph}.}
\label{fig:morph}
\end{wrapfigure}
%
%
%
%
The remaining steps are carried out by the rule \freer, followed by
\empr, with the former disposing the remaining block, thus, completing
the derivation with program \code{listfree} shown in
\autoref{fig:free}.

\subsubsection{Unfolding in the postcondition}
%
Whereas \invkind unfolds predicate instances in a goal's precondition,
a complementary rule \closer (\autoref{fig:induction}) 
performs a similar operation on the goal's postcondition.
The main difference is that instead of performing a case-split
and emitting several subgoals,
\closer \emph{non-deterministically picks} a single clause $k$ from the predicate's definition
(the intuition being that the required case split has already been performed by \invkind).
%
%
%
Upon unfolding, the clause's adapted guard ($[\subst]\sel_k$) and
pure part ($[\subst]\chi_k$) are added to the subgoal's postcondition,
while its spatial part also gets its level tags increased by one
($\bump{[\subst] R_k}{\lev + 1}$), in order to account for the depth
of unfoldings.\footnote{We will elaborate on the control of unfolding
  depth in \autoref{sec:rules}.}
%

To showcase the use of \closer, let us define a new predicate for a
linked null-terminating structure $\lsegt$, which stores in each node the
payload $v$ and $v + 1$:
\begin{equation}
\label{eq:lseg2}
\!\!\!\!\!\!\!
\begin{array}{r@{\ \ }c@{\ \ }r@{\ }c@{\ }l}
\lsegt(x, S) & \eqdef & x = 0 & \wedge & \set{ S \seteq \emptyset ; \emp}
\\[3pt]
& & x \neq 0 & \wedge & \set{\!\! 
\begin{array}{l}
S \seteq \set{v} \union S_1 ;
\\[2pt]
\block{x}{3} \sep x \pts v \sep \angled{x, 1} \pts  v + 1 \sep \angled{x, 2} \pts \nxt \sep \lsegt(\nxt, S_1)
\end{array}
\!\!} 
\end{array}
\end{equation}

We now synthesize an implementation for the following specification,
requiring to morph a regular list $\lseg(x, 0, S)$ to $\lsegt(y, S)$,
both parameterized by the same set $S$:
\begin{equation}
\label{eq:listmorph}
\asn{r \pts 0 \sep \lseg(x, 0, S)}
~\text{\code{void listmorph(loc x, loc r)}}~
\asn{r \pts y \sep \lsegt(y, S)}    
\end{equation}
The derivation starts \mkind, then {\invkind}s $\lseg(x, 0, S)$,
producing two sub-goals.
The first one:
\begin{equation}
\label{eq:morph1}
\set{x, r}; \tent{\asn{S \seteq \emptyset \wedge x = 0 ; r \pts 0}}{\asn{r \pts y \sep \lsegt^0(y, S)}}    
\end{equation}
is easy to solve via \closer, which should pick the first clause from
$\lsegt$'s definition~\eqref{eq:lseg2} (corresponding to $\emp$),
followed by \framer to $\bl{r \pts y}$ in the postcondition. 
The second subgoal, after having read the value of $(x + 1)$, thus
into a program variable $\ccd{nxt2}$ looks as follows:
\begin{equation}
\label{eq:morph2}
\!\!\!\!\!\!\!
\begin{array}{l@{\ \ }l}
\set{x, r, \ccd{nxt2}}; 
\\[2pt]
\asn{S \seteq \set{v} \union S_1 \wedge x \neq 0 ; 
r \pts 0 \sep \block{x}{2} \sep x \pts v \sep \angled{x, 1} \pts \ccd{nxt2}
              \sep \lseg^1(\ccd{nxt2}, 0, S_1)} 
& \leadsto
\\[2pt]
\asn{r \pts y \sep \gbm{\lsegt^0(y, S)}}
\end{array}
\end{equation}
Now, \closer comes to the rescue, by allowing us to unfold the grayed
instance in \eqref{eq:morph2}'s post:
\begin{equation}
\label{eq:morph3}
\!\!\!\!\!\!\!
\begin{array}{l@{\!\!\!\! \!\!}l}
\set{x, r, \ccd{nxt2}}; 
\\ [2pt] 
\asn{S \seteq \set{v} \union S_1 \wedge x \neq 0 ; 
\gbm{r \pts 0} \sep \block{x}{2} \sep x \pts v \sep \angled{x, 1} \pts \ccd{nxt2}
              \sep \gbm{\lseg^1(\ccd{nxt2}, 0, S_1)}} 
& \leadsto
\\[2pt]
\asn{S \seteq \set{v_1} \union S_2; r \pts y \sep
\block{y}{3} \sep y \pts v_1 \sep \angled{y, 1} \pts  v_1 + 1 \sep \angled{y, 2} \pts \nxt_1 \sep \lsegt^1(\nxt_1, S_2)}
\end{array}
\end{equation}
In principle, nothing in the postcondition prevents us from applying
\closer again, unfolding the instance $\bl{\lsegt^1(\nxt_1, S_2)}$
even further. Intuitively, a postcondition with a symbolic heap that
elaborated is less likely to be satisfied, hence we limit the number
of ``telescopic'' unfoldings by enforcing the boundary $\depth$ for
the level tag.
We can now use the \applyr rule, unifying the precondition of the
induction hypothesis~\eqref{eq:listmorph} with the grayed parts in the
goal~\eqref{eq:morph3}, obtaining the following subgoal (notice the
new instance $\bl{\lsegt(y_1, S_1)}$ in the pre with its tag
erased):
\[
\!\!\!
\begin{array}{l@{\ }l}
\set{x, r, \ccd{nxt2}}; 
\\ [2pt] 
\asn{S \seteq \set{v} \union S_1 \wedge x \neq 0 ; 
\gbm{\block{x}{2} \sep x \pts v \sep \angled{x, 1} \pts \ccd{nxt2}}
  \sep r \pts y_1 \sep \lsegt(y_1, S_1)} 
& \leadsto
\\[2pt]
\asn{S \seteq \set{v_1} \union S_2; r \pts y \sep
\gbm{\block{y}{3} \sep y \pts v_1 \sep \angled{y, 1} \pts  v_1 + 1 \sep \angled{y, 2} \pts \nxt_1} \sep \lsegt^1(\nxt_1, S_2)}
\end{array}
\]
The instances of $\lsegt$ in the pre and the post, 
can now be unified via \hunir 
instantiating $\nxt_1$ with $y_1$, 
followed by \hyp on pure parts, instantiating $S_2$ with $S_1$,
and then framed via \framer.
The remaining derivation is done by {\readr}ing from $r$ and $x$,
subsequent disposing of a two-cell block (grayed in the pre) and
allocation of a three-cell block in order to match the grayed block in
the post. Finally, the exact payload for cells of the newly-allocated
3-pointer block is determined by unifying the set assertions in the
pure parts of the pre and post (via \hyp), and then \writer records
the right values to satisfy the constraints imposed for the head of
$\lsegt$-like list by Definition~\eqref{eq:lseg2}.
The resulted synthesized implementation of \code{listmorph} is shown
in \autoref{fig:morph}.


\subsection{Enabling Procedure Calls by Means of Call Abduction}
\label{sec:abduct}

We conclude this overview with one last example---a recursive
procedure for copying a linked list:
\begin{equation}
  \label{eq:listcopy}
  \asn{r \pts x \sep \lseg(x, 0, S)}
  ~\text{\code{void listcopy(loc r)}}~
  \asn{r \pts y \sep \lseg(x, 0, S) \sep \lseg(y, 0, S)}
\end{equation}
To make things more fun, we pass the pointer to the head of the list
via another pointer \code{r}, which is also used to record the result
of the function---an address $y$ of a freshly allocated list copy.
The synthesis begins by using \mkind, producing the function symbol
\begin{equation}
  \label{eq:copyrec}
\!\!\!\!
  \text{\code{void listcopy(loc r')}}: 
  \asn{r' \pts x' \sep \lseg^1(x', 0, S')}
  \asn{r' \pts y' \sep \lseg^1(x', 0, S') \sep \lseg^1(y', 0, S')}
\end{equation}
It follows by \readr (from $r$ into \code{x2}) and \invkind, resulting
in two subgoals, the first of which (an empty list) is trivial.
The synthesis proceeds, reading from \code{x2} into \code{v2} and from
\code{x2 + 1} into \code{nxt2}, so after using \closer (on
$\bl{\lseg^0(x, 0, S)}$) in the post, \hunir and \framer, we reach the
following subgoal:
\begin{equation}
\label{eq:copy1}
\begin{array}{l@{\ \ }l@{\ \ }l}
\set{x, r, \ccd{x2}, \ccd{v2}, \ccd{nxt2}}; 
&
\asn{S \seteq \set{\ccd{v2}} \union S_1 \wedge \ccd{x2} \neq 0 ; 
\gbm{r \pts \ccd{x2} \sep \lseg^1(\ccd{nxt2}, 0, S_1)}} 
& \leadsto
\\[2pt]
&\asn{
r \pts y \sep 
\lseg^1(\ccd{nxt2}, 0, S_2) \sep \lseg^0(y, 0, S)}
\end{array}
\end{equation}

\begin{wrapfigure}[9]{r}{0.4\textwidth}
\setlength{\abovecaptionskip}{0pt}
\vspace{-15pt}
\centering
\begin{minipage}{0.9\linewidth}
{\small{
\centering
\begin{mathpar}
\!\!\!\!\!\!\!\!
\inferrule[\abductr]
{
\fun \eqdef f(\many{x_i}) : \asn{\phi_f; P_f \sep F_f}\asn{\psi_f; Q_f} \in \ctx
\\
F_f~\text{has no predicate instances}
\\
[\subst]P_f = P
\quad
F_f \neq \emp
\quad
F' \eqdef [\subst]F_f
\\
\ctx; \env; \trans{\asn{\phi; F}}{\asn{\phi; F'}}{\prog_1}  
\\
\ctx; \env; \trans{\asn{\phi; P \sep F' \sep R}}{\asn{\slQ}}{\prog_2}  
}
{
\ctx; \env; \trans{\asn{\phi; P \sep F \sep R}}{\asn{\slQ}}{\prog_1; \prog_2}  
}
\end{mathpar}
}}
\end{minipage}
\caption{\abductr rule.}
\label{fig:abductr}
\end{wrapfigure}
At this point of our derivation, we run into an issue. Ideally, we
would like to use the grayed fragment of the goal~\eqref{eq:copy1}'s
precondition, 
to fire the rule \applyr with the spec~\eqref{eq:copyrec},
\ie to make a recursive call on the tail list. 
However, the \eqref{eq:copyrec}'s precondition
requires $r$ to point to the start of that list (\code{nxt2}),
whereas in our case it still points to the start of the original list (\code{x2}).

Any programmer would know a solution to this conundrum: we
have to to write \code{nxt2} into~$r$, in order to provide a
suitable symbolic heap to make a recursive call.  
%
Emitting such a write command is a synthesis sub-goal in itself.
%
To generate such sub-goals, we introduce a novel rule, \abductr, which
is shown in \autoref{fig:abductr} and attempts to \emph{prepare} the
symbolic pre-heap for the recursive call
by adjusting constant-size symbolic footprint to become unifiable
with the recursive calls' precondition.

\begin{wrapfigure}[14]{r}{0.4\textwidth}
\vspace{-16pt}
\setlength{\abovecaptionskip}{+5pt}
\begin{center}
\begin{lstlisting}[style=numbers,xleftmargin=20pt,basicstyle=\footnotesize\ttfamily]
void listcopy (loc r) {
  let x2 = *r;
  if (x2 == 0) { } else {
    let v2 = *x2;
    let nxt2 = *(x2 + 1);
    *r = nxt2;
    listcopy(r);
    let y12 = *r;
    let y2 = malloc(2);
    *y2 = v2;
    *(y2 + 1) = y12;
    *r = y2;
  } }
\end{lstlisting}
\end{center}
\caption{Synthesized \code{listcopy}~\eqref{eq:listcopy}.}
\label{fig:copy}
\end{wrapfigure}
%
%
First, the rule inspects the preconditions of the goal and of the
cadidate callee $\fun$ from $\ctx$, it
tries to split the former into two symbolic sub-heaps, $P_f$ and
$F_f$, such that all predicate instances are contained within
$P_f$, while the rest of the heaplets (\ie, blocks and points-to
assertions) are in $F_f$.
Next, it tries to unify $P_f$ from the function spec with some
sub-heap $P$ from the goal's precondition, finding a suitable
substitution $\subst$, such that $P = [\subst]P_f$.
While doing so, it does not account for the ``remainder''
$[\subst]F_f$, which might not be immediately matched by anything in
the goal's precondition.
In order to make it match, the goal emits, as one of its premises, a
sub-goal
$\ctx; \env; \trans{\asn{\phi; F}}{\asn{\True, F'}}{\prog_1}$, whose
puspose is to synthesize a program $\prog_1$, which will serve as an
impedance matcher between \emph{some} symbolic subheap $F$ from the
original goal's pre and $F' = [\subst]F_f$.\footnote{Our
  implementation is smarter than that: it ensures that $F$ and $F'$
  have the same shape and differ only in pointers' values.}

For instance, in the specification \eqref{eq:copyrec},
$P_f = \bl{\lseg^1(x', 0, S')}$ and $F_f = \bl{r' \pts x'}$, so an
attempt to unify the former with the predicate instance in the
grayed fragment of the goal~\eqref{eq:copy1} results in the
substitution $\subst = [x' \mapsto \ccd{nxt2}, S' \mapsto S_1]$.
Applying it to the remainder of the function spec's pre, we obtain
$F' = [\subst]F_f = \bl{r' \pts \ccd{nxt2}}$.
One of the candidates to the role of $F$ from the goal's precondition
is the heaplet $\bl{r \pts \ccd{x2}}$, so the corresponding subgoal
will be of the form $\set{\ccd{r}, \ldots}; \tent{\asn{\ldots; \ccd{r} \pts
    \ccd{x2}}}{\asn{{r'} \pts \ccd{nxt2}}}$, which will produce
the write \code{*r = nxt2}.
\autoref{fig:copy} shows the eventually synthesized implementation,
with the abduced call-enabling write on line 6.

\section{\logicfull in a Nutshell}
\label{sec:logic}

Having shown \logic in action, we now proceed with giving a complete
set of its inference rules, along with statements of the formal
guarantees \logic provides \wrt synthesized imperative programs.

\begin{figure}[t]
\centering
\begin{tabular}{cc}
\begin{minipage}{0.5\linewidth}
\setlength{\abovecaptionskip}{+1pt}
\[
\!\!\!\!
{\small
\begin{array}{l@{\!\!}r@{\ \ }c@{\ \ }l}
  \text{Variable} & x, y & & \text{Alpha-numeric identifiers}
  \\[2pt]
  \text{Value} & d & & \text{Theory-specific atoms}
  \\[2pt]
  \text{Offset}   & \off & & \text{Non-negative integers}
  \\[2pt]
  \text{Expression} & e & ::= &
   d \mid x \mid e = e \mid e \wedge e \mid \neg e \mid \ldots                             
  \\[2pt]
  \text{Command} &
  c & ::= & 
            \Letz~{x} = \deref{(x + \off)} \mid  
            \deref{(x + \off)} = e \mid 
  \\[2pt]
  & & & \cskip \mid \cerror \mid \cmagic \mid
  \\[2pt]
  & & & 
        \Ifz (e) \set{c}~\Elsez \set{c} \mid
        f(\many{e_i}) \mid c; c
  \\[2pt]
  \text{Type} & t & ::= & 
  \typ{loc} \mid \typ{int} \mid \typ{bool} \mid \typ{set}
  \\[2pt]
  \text{Fun. dict.} & \fctx & ::= & \epsilon \mid \fctx, f~(\many{t_i~x_i})~\set{~c~}
\end{array}
}
\]
\caption{Programming language grammar.}
\label{fig:lang}
\end{minipage}
&
\begin{minipage}{0.5\linewidth}
\[
\!\!\!\!\!\!\!\!\!\!\!\!
{\small
\begin{array}{l@{\ \ }l@{\ \ }c@{\ \ }l}
  \text{Pure assertion} & \phi, \psi, \xi, \chi & ::= &
  e                                                        
  \\[2pt]
  \text{Symbolc heap} & P, Q, R & ::= & 
  \emp \mid \angled{e, \off} \pts e 
  \mid 
  \\[2pt]
  &&&\block{x}{n} \mid p(\many{x_i})                                     
      \mid P \sep Q
  \\[2pt]
  \text{Assertion} & \slP, \slQ & ::= &
  \set{\phi, P}
  \\[2pt]
  \text{Heap predicate} & \indp & ::= &
   p~(\many{x_i})~\many{\pclause{\sel_j}{\chi_j}{R_j}}
  \\[2pt]
  \text{Function spec} & \fun & ::= &
   f~(\many{x_i}) : \asn{\slP} \asn{\slQ}
  \\[2pt]
  \text{Environment} & \env & := & 
  \epsilon \mid \env, x
  \\[2pt]
  \text{Context} & \ctx & := & \epsilon \mid \ctx, \indp \mid \ctx, \fun
\end{array}
}
\]
\caption{\logic assertion syntax.}
\label{fig:logic}
\end{minipage}
\end{tabular}
\end{figure}

The syntax for the imperative language supported by \logic is given in
\autoref{fig:lang}.
The set of values includes at least integers and pointers
(isomorphic to non-negative integers). Expressions include variables,
values, boolean equality checks and additional theory-specific
expressions (\eg, integer or boolean operations).
The command \code{magic} does not appear in runnable code and is
included in the language for the purpose of a deductive synthesis
optimization, which we will explain in \autoref{sec:magic}.
The language of commands does not include loops, which are modelled
via recursive procedure calls ($f(\many{e_i})$).
Notice that for simplicity we do not provide a mechanism to return a
variable from a procedure (so the language is missing the
\code{return} command) and therefore all procedures' return type is
\code{void}. However, the result-returning discipline for a procedure
can be encoded via passing a result-storing additional pointer, as,
\eg, in Example~\eqref{eq:listcopy}.
A \emph{function dictionary} $\fctx$ is simply a list of function
definitions of the form $f~(\many{t_i~x_i})~\set{~c~}$.

The complete syntax of \logic assertions is shown in
\autoref{fig:logic}, and their meaning was explained in detail
throughout \autoref{sec:overview}. 
We only notice here that, syntactically, pure assertions $\phi$,
$\psi$, \etc coincide with the language's expressions $e$.
The lack of the distinction between the two kinds is for the sake of
uniformity and to enable the use of third-party SMT solvers without
committing to a specific first-order logic as an inherent part of
\logic.
We use a simple type system to make sure that the expressions serving
as pure formulae are of type \code{bool}, while also making sure that
set-theoretical operations, such as $\union$, do not leak to the
program level.


\subsection{The Zoo of \logic Rules}
\label{sec:rules}

\autoref{fig:all} presents \emph{all} rules of \logic.
%
%
%
%
Since most of them have already made an appearance in
\autoref{sec:overview}, here, we only elaborate on the new ones, and
highlight some important aspects of their interaction.
It is convenient to split the set of rules into the following six
categories:
  
\begin{enumerate}[label=\textbf{C\arabic*}]
\item \label{ph:toplevel}
  \emph{Top-level rules} are represented by just one rule: \mkind
  (\autoref{fig:all}, bottom right). This rule is only
  applicable at the very first stage of the derivation, and it
  produced a specified symbol $f$, with the specification identical to
  the top-level goal (modulo renaming of variables to avoind name
  capturing conflicts).
  %
  %
  In the case of \emph{several} predicate instances
  $p^{0}(\many{e_i})$ in the goal'pre, the
  \begin{figure}[H]
\centering
\begin{tabular}{c}
\begin{tabular}{ccc}
\begin{minipage}{0.3\linewidth}
{\footnotesize{
\centering
\begin{mathpar}
\!\!\!\!\!\!\!\!\!\!\!\!\!\!\!\!\!\!\!\!\!\!\!\!
\inferrule[\mkind]
{
f \eqdef \text{goal's name}
\\
\many{x_i} \eqdef \text{goal's formals}
\\
P_f \eqdef p^{1}(\many{y_i}) \sep \bump{P}{}
\\
Q_f \eqdef \bump{Q}{}
\\
\fun \eqdef f(\many{x_i}) : \asn{\phi_f; P_f}\asn{\psi_f; Q_f}
\\
\ctx, \fun; \env; \trans{\asn{\phi; p^{0}(\many{y_i}) \sep P}}{\asn{\slQ}}{\prog}  
}
{
\ctx; \env; \trans{\asn{\phi; p^{0}(\many{y_i}) \sep P}}{\asn{\slQ}}{\prog}  
}
\end{mathpar}
}}
\end{minipage}
&
\begin{minipage}{0.3\linewidth}
{\footnotesize{
\begin{tabular}{@{\!\!\!\!\!\!\!\!\!\!\!\!\!}c}
\begin{mathpar}
\inferrule[\empr]
%
{
\ev{\env, \slP, \slQ} = \emptyset
\\
\phi \Rightarrow \psi
}
{
\env; \trans{\asn{\phi; \emp}}{\asn{\psi; \emp}}{\cskip}  
}
\\
\inferrule[\incons]
{\phi \Rightarrow \bot}
{
\env; 
\trans{\asn{\phi ; P}}
{\asn{\slQ}}{\cerror}  
}
\end{mathpar}      
\end{tabular}
}}
\end{minipage}
&
\begin{minipage}{0.3\linewidth}
{\footnotesize{
\centering
\begin{tabular}{@{\!\!\!\!\!\!\!\!\!}c}
\begin{mathpar}
\inferrule[\nilvalr]
{
x \neq 0 \notin \phi
\\
\phi' \eqdef \phi  \wedge x \neq 0
\\
\ctx; \env; \trans{\asn{\phi'; \angled{x, \off} \pts e \sep P}}{\asn{\slQ}}{\prog}
}
{
\ctx; \env; \trans{\asn{\phi;
  \angled{x, \off} \pts e \sep P}}{\asn{\slQ}}{\prog}
}
\end{mathpar}
\\\\
\inferrule[\sleft]
{
\phi \Rightarrow x = y
\\
\env; \trans{[y/x]\asn{\phi; P}}{[y/x]\asn{\slQ}}{\prog}
}
{
\env; \trans{\asn{\phi; P}}{\asn{\slQ}}{\prog}
}      
\end{tabular}
}}
\end{minipage}
\end{tabular}
\\
\begin{tabular}{cc}
\begin{minipage}{0.5\linewidth}
\centering
{\footnotesize{
\begin{mathpar}
\!\!\!\!\!\!\!\!\!\!\!\!
\!\!\!\!\!\!\!\!\!\!\!\!
\inferrule[\starpar]
{
x + \off \neq y + \off' \notin \phi
\\
\phi' \eqdef \phi \wedge (x + \off \neq y + \off')
\\
\ctx; \env; 
\trans{\asn{\phi'; 
    \angled{x, \off} \pts e \sep \angled{y, \off'} \pts e' \sep P}}
{\asn{\slQ}}{\prog}  
}
{
\ctx; \env; 
\trans{\asn{\phi; 
    \angled{x, \off} \pts e \sep \angled{y, \off'} \pts e' \sep P}}
{\asn{\slQ}}{\prog}  
}
\end{mathpar}
}}
\end{minipage}
&
\begin{minipage}{0.5\linewidth}
\centering
{\footnotesize{
\begin{mathpar}
\!\!\!\!\!\!\!\!\!\!\!\!\!\!\!\!\!
\inferrule[\readr]
%
%
{
a \in \gv{\env, \slP, \slQ}
\\
y \notin \vars{\env, \slP, \slQ}
\\
\env \cup \set{y}; \trans{[y/a]\asn{\phi; \angled{x, \off} \pts a \sep P}}{[y/a]\asn{\slQ}}{\prog} 
}
{
\ctx; \env; \trans{\asn{\phi; \angled{x, \off} \pts a \sep P}}
         {\asn{\slQ}}{\Letz y = \deref{(x + \off)}; \prog}  
}
\end{mathpar}
}}
\end{minipage}
\end{tabular}
\\
\begin{tabular}{cc}
\begin{minipage}{0.5\linewidth}
\centering
{\footnotesize{
\begin{mathpar}
\!\!\!\!\!\!\!\!\!\!\!\!\!
\inferrule[\invkind]
{
\indp \eqdef p(\many{x_i})\many{\pclause{\sel_j}{\chi_j}{R_j}}_{j \in
  1\ldots N} \in \ctx
\\
\lev < \depth
\quad
\subst \eqdef [\many{x_i \mapsto y_i}]
\quad
\vars{\many{y_i}} \subseteq \env
\\
\phi_j \eqdef \phi \wedge [\subst]\sel_j \wedge [\subst]\chi_j
\\
P_j \eqdef \bump{[\subst]R_j}{\lev + 1} \sep \bump{P}{}
\\
\forall j \in {1{\ldots}N}, \quad
\ctx; \env; \trans{\asn{\phi_j; P_j}}{\asn{\slQ}}{\prog_j}
\\
\prog \eqdef \Ifz~([\subst]\sel_1) \set{\prog_1}\!\!~\Elsez\! \set{\Ifz~([\subst]\sel_2) \ldots \Elsez~\set{\prog_N}}
}
{
\ctx; \env; \trans{\asn{\phi; P \sep \cached{p^{\lev}(\many{y_i})}}}{\asn{\slQ}}{\prog}  
}
\end{mathpar}
}}
\end{minipage}
&
\begin{minipage}{0.5\linewidth}
\centering
{\footnotesize{
\begin{mathpar}
\!\!\!\!\!\!\!\!\!\!\!\!\!\!\!\!\!\!\!\!\!\!\!\!\!\!\!
\inferrule[\closer]
{
\indp \eqdef p(\many{x_i})\many{\pclause{\sel_j}{\chi_j}{R_j}}_{j \in
  1\ldots N} \in \ctx
\\
\lev < \depth
\\
\subst \eqdef [\many{x_i \mapsto y_i}]
\\
\text{for some}~k, 1 \leq k \leq N
\\
R' \eqdef \bump{[\subst] R_k}{\lev + 1}
\\
\ctx;\env; \trans{\asn{\slP}}
         {\asn{\psi \wedge  [\subst]\sel_k \wedge [\subst]\chi_k; Q \sep R'}}{\prog}  
}
{
\ctx;\env; \trans{\asn{\slP}}
         {\asn{\psi; Q \sep \cached{p^{\lev}(\many{y_i})}}}{\prog}  
}
\end{mathpar}
}}
\end{minipage}
\end{tabular}
\\
\begin{tabular}{cc}
\begin{minipage}{0.5\linewidth}
\centering
{\footnotesize{
\begin{mathpar}
\!\!\!\!\!\!\!\!\!\!\!\!\!\!\!\!\!\!
\inferrule[\abductr]
{
\fun \eqdef f(\many{x_i}) : \asn{\phi_f; P_f \sep F_f}\asn{\psi_f; Q_f} \in \ctx
\\
F_f~\text{has no predicate instances}
\\
[\subst]P_f = P
\\
F_f \neq \emp
\\
F' \eqdef [\subst]F_f
\\
\ctx; \env; \trans{\asn{\phi; F}}{\asn{\phi; F'}}{\prog_1}  
\\
\ctx; \env; \trans{\asn{\phi; P \sep F' \sep R}}{\asn{\slQ}}{\prog_2}  
}
{
\ctx; \env; \trans{\asn{\phi; P \sep F \sep R}}{\asn{\slQ}}{\prog_1; \prog_2}  
}
\end{mathpar}
}}
\end{minipage}
&
\begin{minipage}{0.5\linewidth}
\centering
{\footnotesize{
\begin{mathpar}
\!\!\!\!\!\!\!\!\!\!\!\!\!\!\!\!\!\!\!
\inferrule[\applyr]
{
\fun \eqdef f(\many{x_i}) : \asn{\phi_f; P_f}\asn{\psi_f; Q_f} \in \ctx
\\
R \eqTags [\sigma]P_f
\\
\phi \Rightarrow [\sigma]\phi_f
\\
\phi' \eqdef [\subst]\psi_f
\\
R' \eqdef \bump{[\subst]Q_f}{}
\\
\many{e_i} = [\subst]\many{x_i}
\\
\vars{\many{e_i}} \subseteq \env
\\
\ctx; \env; \trans{\asn{\phi \wedge \phi'; P \sep R'}}{\asn{\slQ}}{\prog}  
}
{
\ctx; \env; \trans{\asn{\phi; P \sep \cached{R}}}{\asn{\slQ}}{f(\many{e_i}); \prog}  
}
\end{mathpar}
}}
\end{minipage}
\end{tabular}
\\
\begin{tabular}{cc}
\begin{minipage}{0.5\linewidth}
{\footnotesize{
\centering
\begin{mathpar}
\!\!\!\!\!\!\!\!\!\!\!\!\!
\inferrule[\allocr]
{
R = \block{z}{n} \sep \Sep_{0 \leq i \leq n}\paren{\angled{z, i} \pts e_i} 
\quad
z \in \ev{\env, \slP, \slQ}
\\
\paren{\set{y} \cup \set{\many{t_i}}} \cap \vars{\env, \slP, \slQ} = \emptyset
\\
R' \eqdef \block{y}{n} \sep \Sep_{0 \leq i \leq n}\paren{\angled{y, i} \pts
  t_i}
\\
\ctx; \env; \trans{\asn{\phi; P \sep R'}}{\asn{\psi; Q \sep R}}{\prog}  
}
{
\ctx; \env; \trans{\asn{\phi; P}}{\asn{\psi; Q \sep \cached{R}}}{\Letz y =
  \malloc{n}; \prog}  
}
\end{mathpar}
}}
\end{minipage}
&
\begin{minipage}{0.5\linewidth}
{\footnotesize{
\centering
\begin{mathpar}
\!\!\!\! \!\!\!\!\!\!\!\!\! \!\!\!\!\!\!\!\!\! \!\!\!\!\!\!\!\!\!
\inferrule[\freer]
{
R = \block{x}{n} \sep \Sep_{0 \leq i \leq n}\paren{\angled{x, i} \pts e_i} 
\\
\vars{\set{x} \cup \set{\many{e_i}}} \subseteq \env
\\
\ctx; \env; \trans{\asn{\phi; P}}{\asn{\slQ}}{\prog}  
}
{
\ctx; \env; \trans{\asn{\phi; P \sep \cached{R}}}{\asn{\slQ}}{\free{n}; \prog}  
}
\end{mathpar}
}}
\end{minipage}
\end{tabular}
\\
\begin{tabular}{cc}
\begin{minipage}{0.7\linewidth}
{\footnotesize{
\centering
\begin{mathpar}
\!\!\!\! \!\!\!\!\!\!\!\!\! \!\!\!\!\!\!\!\!\!
\inferrule[\writer]
%
%
{
\vars{e} \subseteq \env
\\
\env; \trans{{\asn{\phi; \angled{x, \off} \pts e \sep P}}}
         {\asn{\psi; \angled{x, \off} \pts e \sep Q}} {\prog}
}
{
\left.
\env; \asn{\phi; \angled{x, \off} \pts e' \sep P} \leadsto
\asn{\psi; \angled{x, \off} \pts e \sep Q}
~\right|~ {\deref{(x + \off)} = e; \prog}  
}
\end{mathpar}
}}
\end{minipage}
&
\begin{minipage}{0.5\linewidth}
{\footnotesize{
\centering
}}
\end{minipage}
\end{tabular}
\\
\begin{tabular}{cc}
\begin{minipage}{0.5\linewidth}
{\footnotesize{
\centering
\begin{mathpar}
\!\!\!\!\!\!\!\!\!\!\!\!\!\!\!\!\!\!\!\!\!\!\!\!
\inferrule[\hunir]
{
[\subst]R' = R
\\
\gbm{\Frameable}\paren{R'}
\\
\emptyset \neq \dom{\subst} \subseteq \ev{\env, \slP, \slQ}
\\
\env; \trans{\asn{P \sep R}}{[\subst]\asn{\psi; Q \sep R'}}{\prog}
}
{
\env; \trans{\asn{\phi; P \sep \cached{R}}}{\asn{\psi; Q \sep \cached{R'}}}{\prog}  
}
\end{mathpar}
}}
\end{minipage}
&
\begin{minipage}{0.5\linewidth}
{\footnotesize{
\centering
\begin{mathpar}
\!\!\!\!\!\!\!\!\!\!\!\!\!\!\!\!\!\!\!\!\!\!\!\!\!\!\!\!
\inferrule[\framer]
{
\ev{\env, \slP, \slQ} \cap \vars{R} = \emptyset
\\
\gbm{\Frameable}\paren{R'}
\\
\env; \trans{\asn{\phi; P}}{\asn{\psi; Q}}{\prog}  
}
{
\env; \trans{\asn{\phi; P \sep \cached{R}}}{\asn{\psi; Q \sep \cached{R}}}{\prog}  
}
\end{mathpar}
}}
\end{minipage}
\end{tabular}
\\
\begin{tabular}{ccc}
\begin{minipage}{0.33\linewidth}
{\footnotesize{
\centering
\begin{mathpar}
\!\!\!\!\!\!\!\!\!\!\!\!\!\!
\inferrule[\pickr]
{
y \in \ev{\env, \slP, \slQ}
\\
\vars{e} \in \env \cup \gv{\env, \slP, \slQ}
\\
\env; \trans{\asn{\phi; P}}
         {[e/y]\asn{\psi; Q}}{\prog}  
}
{
\env; \trans{\asn{\phi; P}}
         {\asn{\psi; Q}}{\prog}  
}
\end{mathpar}
}}
\end{minipage}
&
\begin{minipage}{0.33\linewidth}
{\footnotesize{
\centering
\begin{mathpar}
\!\!\!\!\!\!\!\!\!\!\!\!\!\!\!\!\!\!\!\!\!\!
\inferrule[\hyp]
{
[\subst]\psi' = \phi'
\\
\emptyset \neq \dom{\subst} \subseteq \ev{\env, \slP, \slQ}
\\
\env; \trans{\asn{\slP}}{[\subst]\asn{\slQ}}{\prog}
}
{
\env; \trans{\asn{\phi \wedge \phi'; P}}{\asn{\psi \wedge \psi'; Q}}{\prog}  
}
\end{mathpar}
}}
\end{minipage}
&
\begin{minipage}{0.33\linewidth}
{\footnotesize{
\centering
\begin{mathpar}
\!\!\!\!\!\!\!\!\!\!\!\!\! \!\!\!\!\!\!\!\!\!
\inferrule[\sright]
{
x \in \ev{\env, \slP, \slQ}
\\
\ctx; \env; \trans{\asn{\slP}}{[e/x]\asn{\psi, Q}}{\prog}
}
{
\ctx; \env; \trans{\asn{\slP}}{\asn{\psi \wedge x = e; Q}}{\prog}
}
\end{mathpar}
}}
\end{minipage}
\end{tabular}
\end{tabular}
\caption{All rules of \logic. \gbm{\text{Grayed parts}} are
  parameters; instantiating them differently yields different rules. 
  %
}
\label{fig:all}
\end{figure}

  rule non-deterministically picks one recursion scheme, whereas other
  predicate instances in $f$'s precondition get ``sealed'' via tag
  erasure $\bump{P}{}$).
  
\item \label{ph:axiom}
  \emph{Terminals} include \empr and \incons.
  These rules conclude a successful derivation
  by emitting \code{skip} or \code{error}, respectively.  

\item \label{ph:universal}
  %
  %
  \emph{Normalization rules} include \nilvalr, \sleft, \starpar, and
  \readr. The role of these rules is to normalize the precondition:
  eliminate ghosts and equal variables, and make explicit the
  assumptions encoded in the spatial part. As we discuss in
  \autoref{sec:inversible}, a characteristic feature of those rules is
  that they cannot cause their subderivation fail.

\item \label{ph:unfolding}
  \emph{Unfolding rules} are the rules targeted specifically
  to predicate instances.
  The four basic unfolding rules are \invkind, \closer, \abductr, and \applyr.
  As hinted before, the rule \invkind picks an instance
  $p^{\lev}(\many{e_i})$ in the goal's pre, and produces a number of
  subgoals, 
  emitting a conditional that accounts for all of $p$'s clauses. \closer acts
  symmetrically on the post. Both rules increment the tag for the
  instances in the substituted clause bodies, 
  and are not applicable for $\lev$, which is equal or greater
  than the set boundary on the depth of unfoldings, defined by the
  global synthesis parameter $\depth$.\footnote{Having the
    derivation/search depth bound by $\depth$ affects the completeness
    but not soundness of \logic.}
  In our practical experience $\depth$ taken to be 
  equal $1$ suffices for most of realistic case studies.
  The rule \abductr prepares the heap for a call application (via
  \applyr), as described in \autoref{sec:abduct}.
  %
  %
  %
  Note that the rules \hunir and \framer have been slightly generalized \wrt to what
  was shown in \autoref{fig:basic-rules} and \autoref{fig:frame-write},
  and now include a parametric premise $\gbm{\Frameable}$.
  If we define $\gbm{\Frameable}$
  to return $\True$ \emph{only} for predicate instances,
  we obtain \emph{unfolding versions} of these rules,
  which deal specifically with predicate instances.

\item \label{ph:flat}
  \emph{Flat} rules deal with with the flat part of the heap
  (which consists of block heaplets and heaplets of the form $x \pts e$).
  They include \allocr, \freer, \writer, 
  as well as the \emph{flat versions} of \hunir and \framer,
  with $\gbm{\Frameable}$ defined to return $\True$ for flat heaplets.
  %

\item \label{ph:pure} \emph{Pure synthesis} rules are responsible for
  instantiating existentials. In \autoref{fig:all} this category is
  represented by three rules. The first one is nondeterministic
  \pickr, which is difficult to implement efficiently.
  In order to make the presentation complete, we include two
  remaining, somewhat redundant but more algorithmic, pure synthesis
  rules introduced in \autoref{sec:pure-synthesis} for the sake of
  efficient search, namely \hyp and \sright.

\end{enumerate}

%
Let us refer to rules that emit a sub-program as \emph{operational}
and to the rest (\ie, to the rules that only change the assertions) as
\emph{non-operational}.
The rules from \autoref{fig:all} form the basis of \logic as a proof
system, allowing for possible extensions for the sake of optimization
or handling pure constraints. We make them intentionally
\emph{declarative} rather than \emph{algorithmic}, which is essential
for establishing the logic's soundness, leaving a lot of freedom for
possible implementations.
Such decisions have to be made, for instance, when engineering an
implementation of \abductr or \pickr.
The algorithmic aspects of \logic, \eg, non-deterministic choice of a
frame or a unifying substitution, are handled by the procedures from
\autoref{sec:algo}.


\subsection{Formal Guarantees for the Synthesized Programs}
\label{sec:formal}

The definition of the operational semantics of the \logic language
(\autoref{fig:lang}) follows the standard RAM model.
\emph{Heaps} (ranged over by $h$) are represented as partial finite maps from
pointers to values, with support for pointer arithmetic (via offsets).
A function call is executed within its own stack frame $(c, s)$,
where $c$ is the next command to reduce and $s$ is a \emph{store}
recording the values of the function's local variables and parameters.
A stack $S$ is a sequence of stack frames, and a \emph{configuration}
is a pair of a heap and a stack.
The small-step operational semantics relates a function dictionary
$\fctx$, and a pair of configurations:
$\opsem{\fctx}{\angled{h, S}}{\angled{h', S'}}$,\footnote{We elide the
  transition rules for brevity; similar rules can by found, \eg, in
  the work by~\citet{Rowe-Brotherston:CPP17}.}
with $\Opsemt$ meaning its reflexive-transitive closure.

Let us denote the \emph{valuation} of an expression $e$ under a store
$s$ as $\denot{e}_s$.
Let $\interp$ range over \emph{interpretations}---mappings from
user-provided predicates $\indp$ to the relations on heaps and vectors
of values.
To formally define the \emph{validity} of Hoare-style specs in \logic,
we use the standard definition of the satisfaction relation
$\satsl{}{}$ as a relation between pairs of heaps and stores,
contexts, interpretations, and \logic assertions without
ghosts.
For instance, the following \logic definitions are traditional for
Separation Logics with interpreted
predicates~\cite{Nguyen-al:VMCAI07,Berdine-al:APLAS05}:
\begin{itemize}
\item $\satsl{\angled{h, s}}{\asn{\phi; \emp}}$ \Iff
$\denot{\phi}_s = \True$ and $\dom{h} = \emptyset$.

\item $\satsl{\angled{h, s}}{\asn{\phi; \block{x}{n}}}$ \Iff
$\denot{\phi}_s = \True$ and $\dom{h} = \emptyset$.

\item $\satsl{\angled{h, s}}{\asn{\phi; \angled{e_1, \off} \pts e_2}}$
  \Iff $\denot{\phi}_s = \True$ and $\dom{h} = \denot{e_1}_s + \off$
  and $h(\denot{e_1}_s + \off) = \denot{e_2}_s$.

\item $\satsl{\angled{h, s}}{\asn{\phi; P_1 \sep P_2 }}$ \Iff
$\exists h_1, h_2, h = h_1 \union h_2$ and
$\satsl{\angled{h_1, s}}{\asn{\phi; P_1}}$ and
$\satsl{\angled{h_2, s}}{\asn{\phi; P_2 }}$.

\item $\satsl{\angled{h, s}}{\asn{\phi; p(\many{x_i})}}$ \Iff
$\denot{\phi}_s = \True$ and
$\indp \eqdef p(\many{x_i})\many{\pclause{\sel_j}{\chi_j}{R_j}} \in \ctx$ and
$\angled{h, \many{\denot{x_i}_s}} \in \interp(\indp)$.
\end{itemize}
\noindent
Therefore, blocks have no spatial meaning, except for serving as an
indicator on the memory fragments that we allocated and can be
disposed.

To equip ourselves for the forthcoming proof of \logic soundness, we
provide a definition of (sized) specification \emph{validity}, which
relies on the notion of a heap size $|h| \eqdef |\dom{h}|$:

%
\begin{definition}[Sized validity]
\label{def:gvalidity}
We say a specification $\ctx; \env; \asn{\slP}~c~\asn{\slQ}$ is
\emph{$n$-valid} \wrt the function dictionary $\fctx$ whenever
for any $h, h', s, s'$ such that
\begin{itemize}

\item $|h| \leq n$,

\item
  $\opsemt{\fctx}{\angled{h, (c, s) \cdot \epsilon}}{\angled{h',
      (\ccd{skip}, s') \cdot \epsilon}}$, and

\item $\dom{s} = \env$ 
and 
$\exists \subst_{\text{gv}} = [\many{x_i \mapsto d_i}]_{x_i \in \gv{\env, \slP,
    \slQ}}$
such that
$\satsl{\angled{h, s}}{[\subst_{\text{gv}}]\slP}$,
\end{itemize}
it is the case that
$\exists \subst_{\text{ev}} = [\many{y_j \mapsto d_j}]_{y_j \in
  \ev{\env, \slP, \slQ}}$, such that
$\satsl{\angled{h', s'}}{[\subst_{\text{ev}} \union
  \subst_{\text{gv}}]}\slQ$



\end{definition}
%
Definition~\ref{def:gvalidity} is rather peculiar in that it defines a
standard Hoare-style soundness \wrt pre/post, while doing
that only for heaps of size smaller than $n$.
This is an important requirement to stage a well-founded inductive
argument for our soundness proof of \logic-based synthesis in the
presence of recursive calls.
By introducing the explicit sizes to the definition of validity, we
can make sure that we only deal with calls on strictly decreasing
subheaps \wrt the heap size when invoking functions (auxiliary ones or
recursive self). 
To make full use of this idea, we add one more auxiliary definition,
stratifying the shape of function dictionaries \wrt their
specifications.

\begin{definition}[Coherence]
\label{def:coh}
A dictionary $\fctx$ is $n$-\emph{coherent} \wrt a context $\ctx$
($\coh{\fctx, \ctx, n}$) \Iff
\begin{itemize}
\item $\fctx = \epsilon$ and $\funs(\ctx) = \epsilon$, or

\item $\fctx = \fctx', f~(\many{t_i~x_i})~\set{~c~}$, and
  $\ctx = \ctx', f~(\many{x_i}) : \asn{\slP} \asn{\slQ}$, and
  $\coh{\fctx', \ctx', n}$, and 
  $\ctx'; \set{\many{x_i}}; \asn{\slP}~c~\asn{\slQ}$ is $n$-valid \wrt
  $\fctx'$, or

\item $\fctx = \fctx', f~(\many{t_i~x_i})~\set{~c~}$, and
  $\ctx = \ctx', f~(\many{x_i}) : \asn{\phi; \bump{P}{} \sep
    p^{1}(\many{e_i})} \asn{\bump{\slQ}{}}$, and $\coh{\fctx', \ctx', n}$, and 
  
  $\ctx; \set{\many{x_i}}; \asn{\bump{P}{} \sep
    p^{1}(\many{e_i})}~c~\asn{\bump{\slQ}{}}$ is $n'$-valid \wrt
  $\fctx$ for all $n' < n$.
\end{itemize}
\end{definition}

That is, coherence is defined inductively on the dictionary/context
shape (regarding specified functions and ignoring predicate
definitions). A possibility of a (single) recursive definition $f$ is
taken into the account in its last option.
In that last clause, recursive calls to the function $f$ from $\fctx$
may only take place on heaps of size \emph{strictly smaller} than $n$,
whereas there is no such restriction for the calls to user-provided
auxiliary functions, that can be invoked on heaps of sizes up to $n$.
%

Soundness of \logicfull is stated by Theorem~\ref{thm:soundness},
which defines validity of the synthesized program \emph{for any} size
of the input heap $n$, assuming that the validity of all recursive
calls of the synthesized programs is only established for $n' < n$,
where the case $n = 0$ means no calls take place and forms the base of
the inductive reasoning.

\begin{theorem}[Soundness of \logic]
\label{thm:soundness}
For any $n$, $\fctx'$, if
\begin{enumerate}[label=\emph{(\roman*)}]

\item\label{s1} $\ctx'; \env; \trans{\asn{\slP}}{\asn{\slQ}}{\prog}$ for a goal named
$f$ with formal parameters~$\env \eqdef \many{x_i}$, and 

\item\label{s2} $\ctx'$ is such that $\coh{\fctx', \ctx', n}$, and

\item\label{s3} for all $p^0(\many{e_i}), \phi; P$,
such that $\asn{\slP} = \asn{\phi; p^0(\many{e_i}) \sep P}$,
taking $\fun \eqdef f(\many{x_i}) : \asn{\phi; p^1(\many{e_i}) \sep
  \bump{P}{}}\asn{\bump{\slQ}{}}$,\\
$\ctx', \fun; \env; \asn{\slP}~\prog~\asn{\slQ}$ is $n'$-valid for
all $n' < n$ wrt.
$\fctx \eqdef \fctx', f~(\many{t_i~x_i})~\set{~\prog~}$,
\end{enumerate}
then $\ctx'; \env; \asn{\slP}~\prog~\asn{\slQ}$ is $n$-valid \wrt
$\fctx$.
\end{theorem}
\begin{proof}
  By the top-level induction on $n$ and by inner induction on the
  structure of derivation
  $\ctx'; \env; \trans{\asn{\slP}}{\asn{\slQ}}{\prog}$.
  We refer the reader to 
  %
  %
  \autoref{sec:appendix}
  for the details.
\end{proof}
Theorem~\ref{thm:soundness} states that a program derived via \logic
constitutes a \emph{valid} spec with its goal (\ie, all writes, reads
and deallocations in it are \emph{safe} \wrt accessing heap pointers),
assuming that recursive calls, if present, are made on reduced
sub-heaps.
The technique we used---allowing for safe calls done only on smaller
sub-heaps---is reminiscent to the one employed in size-change
termination-based analyses~\cite{Lee-al:POPL01}, which ensure that
every infinite sequence of calls would cause infinite descent of some
values, leading to the following result.
\begin{theorem}[Termination of \logic-synthesized programs]
  A program, which is derived via \logic rules from a spec that uses
  only well-founded predicates, terminates.
\end{theorem}
\begin{proof}
  The only source of non-termination is calls. Auxiliary function calls
  cannot be chained infinitely as they erase tags on their post-heaps.
  Every recursive self-call is applicable after opening a well-founded
  instance, and hence is done on a smaller sub-heap, erasing tags on
  its post-heap.
\end{proof}


\section{Basic \logic-Powered Synthesis Algorithm}
\label{sec:algo}

In this section we show how to turn \logic from a declaratively
defined system (\autoref{fig:all}) into a search algorithm with
backtracking for deriving provably correct imperative programs.

\begin{wrapfigure}[5]{r}{0.42\textwidth}
\vspace{-13pt}
\centering
$
{\small{
\begin{array}{r@{\ \ \in\ \ }r@{\ \ }c@{\ \ }l}
\goal & \goals & ::= & \angled{f, \ctx, \env, \asn{\slP}, \asn{\slQ}}
\\[2pt] 
\kont & \konts & \eqdef & \paren{\text{Command}}^{n} \to \text{Command}
\\[2pt]
\deriv & \derivs & ::= & \angled{\many{\goal_i}, \kont}
\\[2pt]
\Rule & \Rules & \eqdef & \goals \partialfn \pow{\derivs}
\end{array}
}}
$
%
\end{wrapfigure}
\paragraph{Encoding Rules and Derivations}
\label{sec:rule-encoding}
The display on the right shows an algorithmic representation of \logic
derivations and rules. To account for the top-level goal (which
mandates the program synthesizer to generate a runnable function), we
include the function name $f$ into the goal $\goal$, 
whose other components are the context $\ctx$, environment $\env$
(initialized with $f$'s formals), pre $\asn{\slP}$ and post
$\asn{\slQ}$.
A successful application or a rule $\Rule$ results in one or more
alternative \emph{sub-derivations} $\deriv_k$. 
Several alternatives arise when the rule exhibits non-determinism
(\eg, due to choosing a sub-heap to a unifying substitution),
and are explored by a search engine one by one, until it finds one that succeeds. 
This is customary for a conversion of a declarative inference system
to an algorithmic one~\cite[Chapter 16]{Pierce:TAPL}.
%
 
In its turn, each sub-derivation is a pair. Its first component
contains zero (if $\Rule$ is a terminal) or more sub-goals, which
\emph{all} need to be solved (think of a conjunction of a rule's
premises).
The second component of the sub-derivation is a \emph{continuation}
$\kont$, which combines the \emph{list} of commands, produced as a
result of solving subgoals, into a final program.
The arity of a continuation (length of a list it accepts) is the same
as a number of sub-goals the corresponding rule emits (typically one
or more). Zero-arity means that the continuation has been produced by
a terminal, and simply emits a constant program. 
For non-operational rules (\eg, \framer),
$\kont \eqdef \lambda [c]. c$. For operational rules $\kont$ typically
prepends a command to the result (\eg, \writer), or generates a
conditional statement (\invkind).
Therefore, the synthesizer procedure constructs the program 
by applying the continuations of rules
that have succeeded earlier, to the resulting programs by their
subgoals, on the ``backwards'' walk of the recursive search, in the
style of logic programming~\cite{Mellish-Hardy:84}.

\paragraph{The algorithm}
\label{sec:algo}
{\setlength{\belowcaptionskip}{-10pt} 
\begin{figure}[t]
{\small{
\centering
\begin{tabular}[t]{c@{\!\!\!\!\!\!\!\!\!\!\!\!\!\!\!\!\!\!
\!\!\!\!\!\!\!\!\!\!\!\!\!\!\!\!\!\!
\!\!\!\!\!\!\!\!\!\!\!\!\!\!\!\!\!\!
\!\!\!\!\!\!\!\!\!\!\!\!\!\!\!\!\!\!
\!\!\!\!\!\!\!\!\!\!\!\!\!\!\!\!\!\!
\!\!\!\!\!\!\!\!\!\!\!\!\!\!\!\!\!\!
\!\!\!\!\!\!\!\!\!\!\!\!\!\!\!\!\!\!
\!\!\!\!\!\!\!\!\!\!\!\!\!\!\!\!\!\!
\!\!\!\!\!\!\!\!\!\!\!\!\!\!\!\!\!\!
\!\!\!\!\!\!\!\!\!\!\!\!
}c}
\begin{minipage}{1.0\linewidth}
\begin{algorithm}[H]
\vspace{-0.5pt}
\SetAlgoNoLine
\Input{Goal $\goal = \angled{f, \ctx, \env, \asn{\slP}, \asn{\slQ}}$}
\Input{List $\trulez$ of available rules to try}
\Result{Program $\prog$, such that\\
$\phantom{abcdef}\ctx; f~(\env)\set{c}; \env \asn{\slP}\prog\asn{\slQ}$}
\Fn{$\synt~(\goal, \trulez)$} { 
    $\tryRules(\trulez, \goal)$
}
\Fn{$\tryRules~(\rulez, \goal)$} { 
  \Switch {$\rulez$} {
    \lCase{$[~]$~\tarr}{
      $\Fail$
    }
    \label{ln:mt-fail}
    \Case{$\Rule :: \rulez'$ ~\tarr}{
      $\sders = \gbm{\filterDervs\left(\wbm{\Rule(\goal)}\right)}$ \;
      \label{opt:tcom}
      \uIf{$\isEmpty(\sders)$}{$\tryRules(\rulez')$}
      \uElse{$\tryAlts(\sders, \Rule, \rulez', \goal)$}
    }
  }
}
\caption{$\synt~(\goal: \goals, \trulez: \Rules^{*})$}\label{proc:syn}
\label{alg:syn}
\end{algorithm}    
\end{minipage}
&
\begin{minipage}{0.7\linewidth}
\SetAlgorithmName{}{}{}
\RestyleAlgo{style}
\setlength{\algoheightrule}{0.0pt}
\setlength{\algotitleheightrule}{0.0pt}
\begin{algorithm}[H]
\vspace{8pt}
\setcounter{AlgoLine}{11}
\SetAlgoNoLine
\Fn{$\tryAlts~(\ders, \Rule, \rulez, \goal)$} {
  \Switch {$\ders$} {
    \lCase{$[~]$~\tarr}{
      \label{opt:inv}$\gbm{\mathbf{if}~\isInvert(\Rule)~
        \mathbf{then}~\Fail~\mathbf{else}\!}\!~\tryRules(\rulez, \goal)\!\!\!\!$}
    \Case{$\angled{\goalz, \kont} :: \ders'$~\tarr}{
      \Switch{$\solve(\goalz, \kont)$} {
        \lCase{$\Fail$~\tarr}{$\tryAlts(\ders', \Rule, \rulez, \goal)$}
        \lCase{$\prog$~\tarr}{
          \label{opt:magic}$\gbm{\mathbf{if}~c = \mathtt{magic}~
            \mathbf{then}~\tryAlts(\ders', \Rule, \rulez,
            \goal)~\mathbf{else}\!}\!~\prog\!\!\!\!$}
      }
    }
  }
} 
\Fn{$\solve~(\goalz, \kont)$} {
  $\progs := [~]$\;
  $\pickRules = \lambda \goal. \gbm{\mathit{phasesEnabled}~?~\nextRules(\goal)~:}~\allrules$\;
  \label{opt:phases}\For{$\goal \leftarrow \goalz; \prog = \synt(\goal, \pickRules(\goal)); \prog \neq \Fail$}{
    $\progs := \progs~+\!\!+~[\prog]$\;
  }
  
  \lIf{$|\progs| < |\goalz|$}{
    $\Fail~\mathbf{else}~\kont(\progs)$
  }
}
\end{algorithm}    
\end{minipage}
\end{tabular}
}}
\end{figure}}

The pseudocode of our synthesis procedure is depicted by
\autoref{alg:syn}.
Let us ignore the grayed fragments in the pseudocode for now and agree
to interpret the code as if they were not present. Those fragments
corresponds to optimizations, which we describe in detail in
\autoref{sec:optimizations}.
The algorithm is represented by four mutually-recursive functions:
%
%
\begin{itemize}
\item $\synt~(\goal, \trulez)$ is invoked initially on a top-level
  goal, with $\trulez$ instantiated with $\allrules$--all rules from
  \autoref{fig:all}. It immediately passes control to the first
  auxiliary function, $\tryRules$.

\item $\tryRules~(\rulez, \goal)$ iterates through the list $\rulez$
  of remaining rules, applying each one to the goal $\goal$. Once an
  application of some rule $\Rule$ succeeds (\ie, it emits a non-empty
  set of alternative sub-derivations $\sders$), those are passed for
  solving to $\tryAlts$. The case when a rule application emits no
  sub-derivations, is considered a search failure, so the next rule is
  tried from the list, until no more rules remains
  (\autoref{ln:mt-fail}), at which point the search fails.

\item $\tryAlts~(\ders, \Rule, \rulez, \goal)$ 
  recursively processes the list of alternative sub-derivations
  $\ders$, generated by the rule $\Rule$. If the list is exhaused
  (\autoref{opt:inv}), $\tryRules$ is invoked to try the rest of the
  rules $\rulez$. Otherwise, $\solve$ is invoked for an alternative to
  solve all its sub-goals $\goalz$ and apply the continuation $\kont$.
  In the case of success (\autoref{opt:magic}), the resulting program
  $\prog$ is returned.

\item $\solve~(\goalz, \kont)$ tries to solve all subgoals given to
  it, by invoking $\synt$ recursively with a suitable (full) list of
  rules, essentially, restarting the search problem ``one level
  deeper'' into the derivation. Unless some of the goals failed, their
  results are combined via~$\kont$.\footnote{The actual implementation
    is more efficient than that and uses a \emph{breakable} loop.}

\end{itemize}

%
The algorithm  implements is a \naive backtracking search that explores the space of all 
valid \logic derivations rooted at the initial synthesis goal.
The search proceeds in a \emph{depth-first} manner: 
it starts from the root (the initial goal)
and always extends the current incomplete derivation 
by applying a rule to its leftmost \emph{open} leaf
(\ie, a leaf that is not a terminal application).
The algorithm terminates when the derivation is \emph{complete}, 
\ie, it has no open leaves.

In our experience, the algorithm implementation terminated on all
inputs we provided.
It seems probable that this can be established by considering the
following tuples, ordered lexicographically, as a termination measure
for a given goal $\goal$:
$\langle$
\#~0- or 1-tagged predicate instances;
\#~ heaplets in pre and post, for which there is no matching one in
the post/pre;
\#~existentials;
\#~''flat'' heaplets;
\#~conjuncts in the pre;
\#~of points-to heaplets, whose disjointness or non-null-ness is not
captured in the pre
$\rangle$.
Notice that each rule from \autoref{fig:all}, except for \invkind and
\closer reduces this value for the emitted sub-goals. Applicability of
those two rules is handled via $\depth$ parameter.

That said, we have not proven termination of $\synt$ and leave it a
conjecture.











\section{Optimizations and Extensions}
\label{sec:optimizations}

The basic synthesis algorithm presented in \autoref{sec:algo} is
a \naive backtracking in the space of all valid \logic derivations.
Note that this is already an improvement over a blind search in the
space of all programs: some incorrect programs are excluded from
consideration a-priori, such as, \eg, programs that read from
unallocated heap locations. In this section we show how to further
prune the search space by identifying unsolvable goals and avoiding
their exploration.

\subsection{Invertible Rules}
\label{sec:inversible}

Our first optimization relies on a well-known fact from proof theory~\cite{Liang-Miller:TCS09}
that certain proof rules are \emph{invertible}:
applying such a rule to any derivable goal produces a goal that is still derivable.
In other words, applying invertible rules eagerly without backtracking
does not affect completeness.
Algorithm~\ref{alg:syn} leverages this fact in \autoref{opt:inv}:
when all sub-derivations of an invertible rule $\Rule$ fail,
the algorithm need not backtrack and try other rules,
since the failure cannot be due to $\Rule$ and must have been caused by a choice made earlier in the search.

%
%
In \logic, the normalization rules---\readr, \starpar, \nilvalr, and \sleft---are invertible.
The effect of these rules on the goal is either to change a ghost into
a program-level variable or to strengthen the precondition; no rule
that is applicable to the original goal can become inapplicable as a
result of this modification, which is confirmed by inspection of all
rules in~\autoref{fig:all}.
%

\subsection{Multi-Phased Search}
\label{sec:phasing}

Among the rules of \logic described in \autoref{sec:logic},
the unfolding rules are focused on transforming (and eventually eliminating)
instances of inductive predicates,
while flat rules are focused on other types of heaplets (\ie, points-to and blocks).
We observe that if the unfolding rules failed to eliminate
a predicate instance from the goal,
there is no point to apply flat rules to that goal.
It is easy to show that the flat rules can neither eliminate predicates from the goal, 
nor enable previously disabled unfolding rules:
the only unfolding rule that matches on flat heaplets is \applyr,
but those heaplets need not be ``prepared'' by the flat rules,
since that's what \abductr is for.

Following this observation, without loss of completeness,
we can split the synthesis process into two phases:
the \emph{unfolding} phase, where flat rules are disabled,
and the \emph{flat} phase, which only starts  when the goal contains no more predicate instances,
and hence unfolding rules are inapplicable.
This optimization is implemented in \autoref{opt:phases} of Algorithm~\ref{alg:syn}.
As a result, some unsolvable goals will be identified early, in the unfolding phase.
For example, the following goal:
\begin{equation*}
\set{y, a, b} ; \tent{\asn{y \pts b \sep a \pts 0 ;
    }} {\asn{y \pts u \sep u \pts 0 \sep \lseg^1(u, 0, S)}}
\end{equation*}
will fail immediately without exploring fruitless transformations on its flat heap,
since no unfolding rules are applicable (assuming $\depth = 1$).

\subsection{Symmetry reduction}
\label{sec:pruning}

Backtracking search often explores all reorderings of
a sequence of rule applications, even if they \emph{commute},
\ie, the order of applications does not change the end sub-goal.
As an examples, consider the following goal:
\begin{equation*}
\set{x, y, a, b} ; \tent{\asn{x \pts a \sep y \pts b \sep a \pts 0
    }} {\asn{x \pts a \sep y \pts b \sep b \pts 0}}
\end{equation*}
Framing out $x \pts a$ and then $y \pts b$, reveals the unsolvable goal
$\teng{\asn{a \pts 0}}{\asn{b \pts 0}}$;
upon backtracking, the \naive search would try the two applications of \framer
in the opposite order, leading to the same result.

We implemented a \emph{symmetry reduction} optimization
to eliminate redundant backtracking of this kind.
To this end, we keep track of the \emph{footprint} of each rule application, 
\ie, the sub-heaps of its goal's pre and post that the application modifies.
This enables us to identify whether two sequential rule applications commute.
Next, we impose a total order on rule applications;
\autoref{opt:tcom} of Algorithm~\ref{alg:syn} rejects a new rule application $\Rule$
if it commutes with an earlier application $\Rule'$ in the current derivation,
but comes before $\Rule'$ in the total order.

\subsection{Early Failure Rules}
\label{sec:magic}

Sometimes on can identify an unsolvable goal by analyzing its pre
and post.
For example, the goal
\begin{equation*}
\set{x, y} ; \tent{\asn{a = 0 ; x \pts a ;
    }} {\asn{a = u \wedge u \neq 0; x \pts u }}
\end{equation*}
is unsolvable because its pure postcondition is logically \emph{inconsistent} with the precondition.
To leverage this observation and eliminate redundant backtracking,
we extend \logic with \emph{failure rules},
which fire when they identify such unsolvable goals.
Each failure rule is a terminal one, so it prevents further
exploration of the (unsolvable) goal.
Unlike other terminals, 
which conclude a successful derivation with \code{skip} or \code{error},
a failure rule emits a special \emph{spurious program} \code{magic}.
Algorithm~\ref{alg:syn} intercepts any appearance of \code{magic} in \autoref{opt:magic},
and treats the derivation as a unsuccessful.
All failure rules are also invertible,
hence the effect is to backtrack an application of an earlier rule.

Our set of failure rules is shown in \autoref{fig:fail}.
The rule \postinr identifies a goal where the the pure postcondition is inconsistent
with the precondition.
This is safe because during the derivation both assertions can only become stronger 
(as a result of unfolding rules);
also, even if the postcondition still contains existentials,
no instantiation of those existentials can produce a formula that is consistent with
(let alone implied by) the precondition.
The rule \postinl fires on a goal where the pure postcondition (which is free of existentials)
is not implied by the precondition;
this rule only applies when the precondition is free of predicate applications,
and hence its pure part cannot be strengthened any further.
The rule \heapunrf fires when the spatial pre and post contian only points-to heaplets,
but the \emph{left-hand sides} of these heaplets cannot be unified;
in this case neither \hunir nor \writer can make the heaplets match.
Note that it is important for completeness that failure rules are
checked \emph{after} \incons: if the pure precondition is
inconsistent, the derivation should not fail, but should instead emit
\code{error}.

\begin{figure}[t]
\setlength{\abovecaptionskip}{-2pt}
\setlength{\belowcaptionskip}{-10pt}
\centering
\begin{tabular}{c}
\begin{tabular}{c@{\ \ }c@{\ \ }c}
\begin{minipage}{0.3\linewidth}
{\footnotesize{
\centering
\begin{mathpar}
\!\!\!\!\!\!\!\!\!\!
\inferrule[\postinr]
{
\phi \wedge \psi \Rightarrow \bot
}
{
\ctx; \env;  \trans{\asn{\phi; P}} {\asn{\psi, Q}}{\cmagic}  
}
\end{mathpar}
}}
\end{minipage}
&
\begin{minipage}{0.3\linewidth}
{\scriptsize{
\centering
\begin{mathpar}
\!\!\!\!\!\!\!\!\!\!\!\!
\inferrule[\postinl]
{
P~\text{has no pred. instances}
\\
\ev{\env, \slP, \slQ} = \emptyset 
\quad
\neg\paren{\phi \Rightarrow \psi}
}
{
\ctx; \env;  \trans{\asn{\phi; P}} {\asn{\psi, Q}}{\cmagic}  
}
\end{mathpar}
}}
\end{minipage}
&
\begin{minipage}{0.3\linewidth}
{\scriptsize{
\centering
\begin{mathpar}
\inferrule[\heapunrf]
{
P, Q~\text{have no pred. instances or blocks}
\\
\mathsf{unmachedHeaplets}(P, Q)
}
{
\ctx; \env;  \trans{\asn{\phi, P}} {\asn{\psi, Q}}{\cmagic}  
}
\end{mathpar}
}}
\end{minipage}
\end{tabular}
\end{tabular}
\caption{Failure rules.}
\label{fig:fail}
\end{figure}

\subsection{Extensions}
\label{sec:extensions}

We wrap up this section with a description of two \logic extensions
used by our implementation in order to expand the class of programs it can synthesize.

\subsubsection{Auxiliary Functions and \emph{\applyr} Rule}

The presented in \autoref{fig:all} version of \applyr, which erases the
tags from the callee's post, hurts the framework's completeness. While
it is not unsafe to employ the predicates from the procedure call's
postcondition in further procedure calls, one cannot ensure that they
denote ``smaller heaplets'' and thus that the program
terminates.\footnote{A similar issue is reported in the work
  by~\citet{Rowe-Brotherston:CPP17} on verifying termination of
  procedural programs.}
To circumvent this limitation for particularly common scenarios, we
had to introduce one divergence between \tool and \logic as shown in
\autoref{fig:all}.
Specifically, in the tool, we implemented support for
\emph{stratified} \emph{chained} auxiliary function calls, \ie, calls
on a heap resulted from another, preceding, call (\eg, in flatten
w/append and insertion sort).
To allow for them, we had to implement slightly different versions of
\mkind/\applyr rules, with an additional alternative in their
premises, which instead of \emph{erasing} tags of the corresponding
part of the post ($\bump{Q_f}{}$) would \emph{increment} them:
$\bump{Q_f}{\bullet + 1}$.
This would prevent a call of the \emph{same} function on the resulting
heap fragment, but would enable calls of auxiliary functions, whose
specs' pre feature predicates with matching \emph{higher-level} tags.
Eventually, due to incrementation, no applicable functions would have
left in the context~$\ctx$. While this extension is unlikely to break
the \logic soundness and termination results (due to the limit of
chained applications, enabled by growing tags), it would require us to
generalize the well-foundedness argument in~\autoref{sec:formal}, and
in the interest of time we did not carry out this exercise.


\subsubsection{Branch Abduction}

\begin{wrapfigure}[6]{r}{0.4\textwidth}
\setlength{\abovecaptionskip}{0pt}
\vspace{-15pt}
\centering
\begin{minipage}{0.9\linewidth}
{\small{
\centering
\begin{mathpar}
\!\!\!\!\!\!\!\!\!\!
\inferrule[\branchr]
{
\ctx; \env;  \trans{\asn{\phi \wedge \psi; P}} {\asn{\slQ}}{c_1}
\\
\ctx; \env;  \trans{\asn{\phi \wedge \neg\psi; P}} {\asn{\slQ}}{c_2}  
}
{
\ctx; \env;  \trans{\asn{\phi; P}} {\asn{\slQ}}{\Ifz~(\psi)~c_1~\Elsez~c_2}  
}
\end{mathpar}
}}
\end{minipage}
\caption{\branchr rule.}
\label{fig:branchr}
\end{wrapfigure}
As described, \logic rules only emit conditional statements when unfolding an inductive 
predicate instance in the goal's precondition via \invkind.
This prevents our framework from synthesizing some useful functions,
in particular, those that branch on the content on a data structure rather than its shape.
This source of incompleteness can easily be mitigated by adding a rule \branchr (\autoref{fig:branchr}),
which is always applicable and generates a conditional statement with a non-deterministically chosen guard.
Of course, in practice picking the guard blindly is not feasible.
Our implementation employs a variation of a popular technique in program synthesis,
known as \emph{condition} (or \emph{branch}) \emph{abduction}%
~\cite{Alur-al:TACAS17,Kneuss-al:OOPSLA13,Leino-Milicevic:OOPSLA2012,Polikarpova-al:PLDI16}.
Instead of emitting conditionals eagerly,
branch abduction tries to detect when the current program under consideration
is a promising candidate for becoming a branch of a conditional,
and then \emph{abduces} a guard that would make this branch satisfy the specification.

The \logic variation of branch abduction piggy-backs on the failure rule \postinl (\autoref{sec:magic}),
which detect a goal whose pure postcondition $\psi$ has no existentials 
and does not follow from the precondition $\phi$.
Instead of rejecting this goal as unsolvable,
branch abduction searches a small set of pure formulas 
(all atomic formulas over program variables in $\Gamma$);
if it can find a formula $\phi'$ such that $\phi \wedge \phi' \Rightarrow \psi$,
it abduces $\phi'$ as the branch guard for the current goal.

For a simple example, consider the following synthesis goal
that corresponds to computing a lower bound of two integers $x$ and $y$: 
\begin{equation*}
\set{x, y} ; \tent{\asn{r \pts 0}} {\asn{m \le x \wedge m \le y; r \pts m }}
\end{equation*}
We first apply \pickr with $[m \mapsto x]$, arriving at the goal
$\tent{\asn{r \pts 0}} {\asn{x \le x \wedge x \le y; r \pts x }}$.
At this point, the postcondition is invalid, 
so branch abduction fires and infers a guard $x \le y$ for the current derivation.
It also emits a new goal for synthesizing the \code{else} branch,
where the negation of the abduced condition is added to the precondition:
$\tent{\asn{\neg(x \le y) ; r \pts 0}} {\asn{m \le x \wedge m \le y; r \pts m}}$.
If the synthesis of the \code{else} branch succeeds, 
the two branches are joined by a conditional;
otherwise branch abduction fails.

\newcommand{\maxtime}{40\xspace}
\newcommand{\supersec}{four\xspace}
\newcommand{\tototal}{8\xspace}
\newcommand{\importantOpt}{invertible rules\xspace}
\newcommand{\toimportant}{5\xspace}
\newcommand{\xnatural}{six\xspace}

\section{Implementation and Evaluation}
\label{sec:evaluation}

We implemented \logic-based synthesis as a tool, called \tool, in
Scala, using \tname{Z3}~\cite{deMoura-Bjorner:TACAS08} as the back-end
SMT solver via the~\tname{ScalaSMT}
library~\cite{Cassez-Sloane:Scala17}.\footnote{The tool sources can be
  found at \url{https://github.com/TyGuS/suslik}.}
%
%
%
We evaluated our implementation with the goal of answering the following research questions:
\begin{enumerate}
\item \emph{Generality:} Is \tool general enough to synthesize a range
  of nontrivial programs with pointers?
\item \emph{Utility:} How does the size of the inputs required by \tool compare to the size of the generated programs?
Does \tool require any additional hints apart from pre- and post-conditions?
What is the quality of the generated programs? 
\item \emph{Efficiency:} Is it efficient? 
What is the effect of optimizations from \autoref{sec:optimizations} on synthesis times?
\item \emph{Comparison with existing tools:} How does \tool fare in
  comparison with existing tools for synthesizing heap-manipulating
  programs, specifically,
  \tname{ImpSynt}~\cite{Qiu-SolarLezama:OOPSLA17}?
\end{enumerate}

\begin{table}[t]
\begin{center}
\footnotesize
\begin{tabular}{@{} r|c| cc| cccccc| c @{}}
\head{Group} & \head{Description} & \head{Code} & \head{Code/Spec} & \head{Time} & \head{T-phase} & \head{T-inv} & \head{T-fail} & \head{T-com} & \head{T-all} & \head{T-IS} \\	

\hhline{===========}

\multirow{2}{*}{\parbox{1cm}{\center{Integers}}} & swap two & 12 & 0.9x & $<0.1$ & $<0.1$ & $<0.1$ & $<0.1$ & $<0.1$ & $<0.1$ &  \\
 & min of two\textsuperscript{2} & 10 & 0.7x & 0.1 & 0.1 & 0.1 & $<0.1$ & 0.1 & 0.2 &  \\
\hline\multirow{9}{*}{\parbox{1cm}{\center{Linked List}}} & length\textsuperscript{1,2} & 21 & 1.2x & 0.4 & 0.9 & 0.5 & 0.4 & 0.6 & 1.4 & 29x \\
 & max\textsuperscript{1} & 27 & 1.7x & 0.6 & 0.8 & 0.5 & 0.4 & 0.4 & 0.8 & 20x \\
 & min\textsuperscript{1} & 27 & 1.7x & 0.5 & 0.9 & 0.5 & 0.4 & 0.5 & 1.2 & 49x \\
 & singleton\textsuperscript{2} & 11 & 0.8x & $<0.1$ & $<0.1$ & $<0.1$ & $<0.1$ & $<0.1$ & $<0.1$ &  \\
 & dispose & 11 & 2.8x & $<0.1$ & $<0.1$ & $<0.1$ & $<0.1$ & $<0.1$ & $<0.1$ &  \\
 & initialize & 13 & 1.4x & $<0.1$ & 0.1 & 0.1 & $<0.1$ & 0.1 & $<0.1$ &  \\
 & copy\textsuperscript{3} & 35 & 2.5x & 0.2 & 0.3 & 0.3 & 0.1 & 0.2 & - &  \\
 & append\textsuperscript{3} & 19 & 1.1x & 0.2 & 0.3 & 0.3 & 0.2 & 0.3 & 0.7 &  \\
 & delete\textsuperscript{3} & 44 & 2.6x & 0.7 & 0.5 & 0.3 & 0.2 & 0.3 & 0.7 &  \\
\hline\multirow{3}{*}{\parbox{1cm}{\center{Sorted list}}} & prepend\textsuperscript{1} & 11 & 0.3x & 0.2 & 1.4 & 83.5 & 0.1 & 0.1 & - & 48x \\
 & insert\textsuperscript{1} & 58 & 1.2x & 4.8 & - & - & - & 5.0 & - & 6x \\
 & insertion sort\textsuperscript{1} & 28 & 1.3x & 1.1 & 1.8 & 1.3 & 1.2 & 1.2 & 74.2 & 82x \\
\hline\multirow{5}{*}{\parbox{1cm}{\center{Tree}}} & size & 38 & 2.7x & 0.2 & 0.3 & 0.2 & 0.2 & 0.2 & 0.3 &  \\
 & dispose & 16 & 4.0x & $<0.1$ & $<0.1$ & $<0.1$ & $<0.1$ & $<0.1$ & $<0.1$ &  \\
 & copy & 55 & 3.9x & 0.4 & 49.8 & - & 0.8 & 1.4 & - &  \\
 & flatten w/append & 48 & 4.0x & 0.4 & 0.6 & 0.5 & 0.4 & 0.4 & 0.6 &  \\
 & flatten w/acc & 35 & 1.9x & 0.6 & 1.7 & 0.7 & 0.5 & 0.6 & - &  \\
\hline\multirow{3}{*}{\parbox{1cm}{\center{BST}}} & insert\textsuperscript{1} & 58 & 1.2x & 31.9 & - & - & - & - & - & 11x \\
 & rotate left\textsuperscript{1} & 15 & 0.1x & 37.7 & - & - & - & - & - & 0.5x \\
 & rotate right\textsuperscript{1} & 15 & 0.1x & 17.2 & - & - & - & - & - & 0.8x \\
\hline

\hhline{===========}
\multicolumn{10}{l}{}
\\[-8pt]
\multicolumn{10}{l}{\textsuperscript{1} From~\cite{Qiu-SolarLezama:OOPSLA17}\quad
\textsuperscript{2} From~\cite{Leino-Milicevic:OOPSLA2012}\quad
\textsuperscript{3} From~\cite{Qiu-al:PLDI13}}
\vspace{-15pt}
\end{tabular}
\end{center}
\caption{
Benchmarks and \tool results.
For each benchmark, we report the 
size of the synthesized \head{Code} (in AST nodes) and the ratio \head{Code/Spec} of code to specification;
as well as synthesis times (in seconds):
with all optimizations enabled (\head{Time}), 
without phase distinction (\head{T-phase}),
without invertible rules (\head{T-inv}),
without early failure rules (\head{T-fail}),
without the commutativity optimization (\head{T-com}),
and without any optimizations (\head{T-all}).
\head{T-IS} reports the ratio of synthesis time in \tname{ImpSynt} to \head{Time}.
``-'' denotes timeout of \timeout seconds.
}
\label{fig:evaluation}
\end{table}

\subsection{Benchmarks}
\label{sec:benchmarks}

In order to answer these questions, 
we assembled a suite of \exCount programs listed in \autoref{fig:evaluation}.
The benchmarks are grouped by the main data structure they manipulate: 
integer pointers,
singly linked lists,
sorted singly linked lists,
binary trees,
and binary search trees.

To facilitate comparison with existing work,
most of the programs are taken from the literature on synthesis and verification of heap-manipulating programs:
the \tname{ImpSynt} synthesis benchmarks~\cite{Qiu-SolarLezama:OOPSLA17},
the \tname{Jennisys} synthesis benchmarks~\cite{Leino-Milicevic:OOPSLA2012},
and the \tname{Dryad} verification benchmarks~\cite{Qiu-al:PLDI13}.
We manually translated these benchmarks into the input language of \tool,
taking care to preserve their semantics.
\tname{Dryad} and \tname{ImpSynt} use the \tname{Dryad} dialect of separation logic as their specification language,
hence the translation in this case was relatively straightforward.
As an example, consider an \tname{ImpSynt} specification and its \tool
equivalent in \autoref{fig:spec-comparison}.
The ``\code{??}'' are part of the \tname{ImpSynt} spec language,
denoting unknown holes to be filled by the synthesizer. 
The main difference between the two pre-/post-condition pairs
is that the \tname{Dryad} logic supports recursive functions such as \T{len}, \T{min}, and \T{max};
in \tool this information is encoded in more traditional SL style:
by passing additional ghost parameters to the inductive predicate \T{srtl}. 
The extra precondition $0 \le n \wedge 0 \le k \wedge k \le 7$ in
\tool 
corresponds to implicit axioms in \tname{ImpSynt}
(in particular, the condition on $k$ is due to its encoding of list
elements as unsigned 3-bit integers---there is no such restriction in \tool).
In addition to benchmarks from the literature, 
we also added several new programs that show-case interesting features
of \tool.

\begin{figure}[t]
\setlength{\abovecaptionskip}{0pt}
\setlength{\belowcaptionskip}{-10pt}
\begin{minipage}{.5\textwidth}
\begin{lstlisting}[basicstyle=\footnotesize\ttfamily,morekeywords={requires,ensures,return,old}]
loc srtl_insert(loc x, int k)
requires srtl(x)
ensures srtl(ret) /\ 
  len(ret) == old(len(x)) + 1 /\
  min(ret) == (old(k) < old(min(x)) 
    ? old(k) : old(min(x))) /\
  max(ret) == (old(max(x)) < old(k) 
    ? old(k) : old(max^(x)))
{
  if (cond(1)) {
    loc ?? := new;
    return ??;
  } else {
    statement(1);
    loc ?? := srtl_insert(??, ??);
    statement(1);
    return ??;
  }
}
\end{lstlisting}        
\end{minipage}%
\begin{minipage}{0.5\textwidth}
\begin{lstlisting}[basicstyle=\footnotesize\ttfamily]
{ 
  0 <= n /\ 0 <= k /\ k <= 7 ; 
  ret :-> k ** srtl(x, n, lo, hi) 
}
void srtl_insert(loc x, loc ret)
{ 
  n1 == n + 1 /\ 
  lo1 == (k <= lo ? k : lo) /\ 
  hi1 == (hi <= k ? k : hi) ; 
  ret :-> y ** srtl(y, n1, lo1, hi1) 
}
\end{lstlisting}
\end{minipage}
\caption{(left) The \tname{ImpSynt} input for the sorted list insertion;
(right) The \tool input for the same benchmark.}\label{fig:spec-comparison}
\end{figure}


\subsection{Results}
\label{sec:results}

Evaluation results are summarized in \autoref{fig:evaluation}.
All experiments were conducted on a commodity laptop
(2.7 GHz Intel Core i7 Lenovo Thinkpad with 16GB RAM).

\subsubsection{Generality and Utility}

Our experiment confirms that \tool is capable of synthesizing programs
that manipulate a range of heap data structures, including nontrivial
manipulations that require reasoning about both the shape and the
content of the data structure, such as insertion into a binary search tree.
We manually inspected all generated solutions, as well as their
accompanying \logic derivations, and confirmed that they are indeed
correct.\footnote{In the future, we plan to output \logic derivations
  as SL proofs, checkable by a third-party system such as
  VST~\cite{Appel:ESOP11}.}
Perhaps unsurprisingly, some of the solutions were not entirely intuitive:
as one example, the synthesized version of list copy, in a bizarre yet valid move,
\emph{swaps the tails} of the original list and the copy at each recursive
call!
%
%

Two of the programs in \autoref{fig:evaluation} make use of auxiliary functions:
``insertion sort'' calls ``insert'',
and ``tree flatten w/append'' calls the ``append'' function on linked lists.
The specifications of auxiliary functions have to be supplied by the user
(while their implementations can, of course, be synthesized independently).
Alternatively, tree flattening can be synthesized without using an auxiliary function,
if the user supplies an additional list argument that plays the role of an accumulator
(see ``tree flatten w/acc'').
As such, \tool shares the limitation of all existing synthesizers for recursive functions:
they require the initial synthesis goal to be inductive,
and do not try to discover recursive auxiliary functions
(which is a hard problem, akin to lemma discovery in theorem proving).


For simple programs specification sizes are mostly comparable with the size of the synthesized code, 
whereas more complex benchmarks is where declarative specifications really shine:
for example, for all Tree programs, the specification is at most half the size of the generated code.
Three notable outliers are ``prepend'', ``rotate left'', and ``rotate right'',
whose implementations are relatively short, 
while the specification we inherited from \tname{ImpSynt} describes the effects of the functions on the minimum and maximum of the list/tree.
Note that the specification sizes we report exclude the definitions of inductive predicates,
which are reusable, and are shared between the benchmarks.

\subsubsection{Efficiency}

\tool has proven to be efficient in synthesizing a
variety of programs: all \exCount benchmarks are synthesized within \maxtime seconds, 
and all but \supersec of them take less than a second.

In order to assess the impact on performance of various optimizations described in \autoref{sec:optimizations},
\autoref{fig:evaluation} also reports synthesis times with each optimization \emph{disabled}:
the column \head{T-phase} corresponds to eliminating the distinction between phases;
\head{T-inv} corresponds to ignoring rule invertibility;
\head{T-fail} corresponds to dropping all failure rules;
\head{T-com} corresponds to disabling the symmetry reduction; 
finally, \head{T-all} corresponds to a variant of \tool with \emph{all} the above optimizations disabled.
The results demonstrate the importance of optimizations for nontrivial programs:
\tototal out of \exCount benchmarks time out when all optimizations are disabled.
The simpler benchmarks (\eg,~\code{swap}) do not benefit from the optimizations at all,
since they do not exhibit a lot of backtracking.
At the same time, all three BST benchmarks time-out as a result of disabling even a single optimization.

\subsubsection{Comparison with Existing Synthesis Tools}
\label{sec:comparison}

We compare \tool with the most closely related prior work on \tname{ImpSynt}~\cite{Qiu-SolarLezama:OOPSLA17}.
Out of the 14 benchmarks from~\cite{Qiu-SolarLezama:OOPSLA17} successfully synthesized by \tname{ImpSynt},
we excluded 5 that are not structurally recursive (4 of them use loops, and \code{bst_del_root} uses non-structural recursion);
the other 9 were successfully synthesized by \tool.
The \emph{qualitative difference} in terms of the required user input is immediately obvious 
from the representative example in \autoref{fig:spec-comparison}:
in addition to the declarative specification,
\tname{ImpSynt} requires the user to provide an implementation \emph{sketch},
which fixes the control structure of the program, the positions of function calls, and the number of other statements.
These additional structural constraints are vital for reducing the size of the search space in \tname{ImpSynt}.
Instead, \tool prunes the search space by leveraging the structure inherent in separation logic proofs,
allowing for more concise, purely declarative specifications.

Despite the additional hints from the user, \tname{ImpSynt} is also \emph{less efficient}:
as shown in the column \head{T-IS} of \autoref{fig:evaluation},
on 6 out of 9 common benchmarks, \tname{ImpSynt} takes at least an order of magnitude longer than \tool,
even though the \tname{ImpSynt} experiments were conducted on a 10-core server with 96GB of RAM.
 
On the other hand, \tname{ImpSynt} is \emph{more general} than \tool in that it can synthesize
both recursive and looping programs.
We discuss this and other limitations of \tool in more detail in \autoref{sec:discussion}.


\section{Limitations and Discussion}
\label{sec:discussion}

There are a number of known limitations of \logicfull and \tool
design.

\logic does no support synthesis of programs with \code{while}-loops,
as this would require discovering loop invariants, 
which significantly increases the search space; instead, the
initial goal is considered as inductive and handled via \mkind rule.
As of now, \logic does not allow for mutually-recursive inductive
predicates. While not impossible in principle, this would require us
to explore advanced techniques for inductive proofs~\cite{Ta-al:FM16}
and also generalize the use of tags; we plan to look into this in the
future.
By limiting the number of unfoldings, via \invkind and \closer rules,
via $\depth$, we circumvent a commonly known decidability problem of
solving entailments in the presence of general inductive
predicates~\cite{Antonopoulos-al:FOSSACS04}, but this also prevents
some non-unreasonable and perfectly specifiable in \logic programs
from being synthesized, \eg, allocating a large constant-size list.

Currently, \logic and \tool cannot automatically synthesize programs
that are not structurally recursive \wrt some inductive predicate,
such as, for instance, merging sorted lists or Merge-Sort.
%
%
One approach to mitigate this limitation is to prove termination
by showing that each recursive call decreases the value of a custom \emph{termination metric};
this technique is used in several automated verifiers and 
synthesizers~\cite{Leino:ICSE04,Polikarpova-al:PLDI16}.
A termination metric maps the tuple of function's arguments into
an element of a set with a pre-defined well-founded order
(usually a tuple of natural numbers),
and can be either provided by the user or inferred by the synthesizer.

Some of \logic limitations are inherent for Separation Logics in
general: SLs are known to work well with disjoint tree-like linked
structures, and programs, whose recursion scheme matches the data
definition, but not so well with ramified data structures, \eg,
graphs.
To address those, one could integrate a more powerful, \emph{ramified}
version of \framer rule~\cite{Hobor-Villard:POPL13} into \logic, but
this would likely require more hints from the user, thus, reducing the
utility of the approach.

\section{Related Work}
\label{sec:related}


There are two main directions in the area of program synthesis: 
synthesis from informal descriptions (such as examples, natural language, or hints)
\cite{Albarghouthi-al:CAV13,Polozov-Gulwani:OOPSLA15,Osera-Zdancewic:PLDI15,Feser-al:PLDI15,Smith-Albarghouthi:PLDI16,Feng-al:PLDI17,Yaghmazadeh-al:OOPSLA17,Murali-al:ICLR18}
and synthesis from formal specifications.
We will only discuss the more relevant latter direction. 
The goal of this type of program synthesis is to obtain a \emph{provably correct} program.

In this area, there is a well-known trade-off between three
dimensions: how complex the synthesized programs are, how strong the
correctness guarantees are, and how much input is required form the
user. On one end of the spectrum there are interactive
synthesizers~\cite{Itzhaky-al:OOPSLA16,Delaware-al:POPL15}, which can
be very expressive and provide strong guarantees, but the user is
expected to guide the synthesis process (although, usually, with aid
of dedicated proof \emph{tactics}).
%
%
On the other end, there is fully automated synthesis for loop- and
recursion-free programs over simple domains, like arithmetic and
bit-vectors~\cite{Gulwani-al:PLDI11,Alur-al:TACAS17}. Our work lies in
the middle of this spectrum, where synthesis is automated but programs
are more expressive.

In the presence of loops or recursion, verifying candidates becomes
nontrivial. Synthesizers like
\tname{Sketch}~\cite{Solar-Lezama:STTT13} and
\tname{Rosette}~\cite{Torlak-Bodik:PLDI14} circumvent this problem by
resorting to bounded verification, which only provides restricted
guarantees and has scalability issues due to path explosion. In
contrast, our work relies on unbounded deductive verification.
%

Among synthesis approaches that use unbounded verification,
synthesizers like \tname{Leon}~\cite{Kneuss-al:OOPSLA13} and
\tname{Synquid}~\cite{Polikarpova-al:PLDI16} focus on pure functional
(recursive) programs, which are an easier target for unbounded
verification. Proof-theoretic synthesis~\cite{Srivastava-al:POPL10} is
capable of synthesizing imperative programs with loops and arrays, but
no linked structures; they pioneered the idea of synthesizing
provably-correct programs by performing symbolic (SMT-based) search
over programs \emph{and their verification conditions} simultaneously.

Finally, the two pieces of prior work that are most closely related to
ours in terms of scope are
\tname{Jennisys}~\cite{Leino-Milicevic:OOPSLA2012} and Natural
Synthesis~\cite{Qiu-SolarLezama:OOPSLA17}, both of which generate
provably-correct heap-manipulating programs. Both of them are
essentially instances of proof-theoretic synthesis with a program
logic for reasoning about the heap. To that end, \tname{Jennisys} uses
the \tname{Dafny} verifier~\cite{Leino:ICSE04}, which supports
expressive yet undecidable specifications, and often requires hints
from the user, so in practice the tool doesn't scale to complex
examples (for example, none of their benchmarks performs mutation).
Natural Synthesis uses
\tname{Dryad}~\cite{Madhusudan-al:POPL12,Qiu-al:PLDI13}, a decidable
program logic for reasoning about heap-manipulating programs. The
downside of this approach is that whole-program symbolic search
doesn't scale to larger programs; to mitigate this, they require the
user to provide sketches with substantial restrictions on the
structure of the program. Our approach does not require sketches (but
on the other hand, we do not support loops).


The recent tool \tname{FootPatch} by~\citet{Tonder-LeGoues:ICSE18} is
very close in its methods and goals to \tool.
\tname{FootPatch} builds on
\tname{Infer}~\cite{Calcagno-Distefano:NFM11}, an open-source SL-based
static analyzer by Facebook, using it for \emph{automated program
  repair}.
It takes the intermediate assertions, provided by \tname{Infer} for
programs with bugs, such as resource and memory leaks, and null
dereferences, and constructs additive patches based on the observed
discrepancy. In this, it acts similarly to our \abductr rule.
\tname{FootPatch} does not synthesize patches that would involve
recursion or complex control flow.


Instead of whole-program symbolic search, like in proof-theoretic
synthesis, our work follows the tradition of \emph{deductive
  synthesis}, \ie, backtracking search in the space of program
derivation composed of synthesis rules, which gradually transform a
specification into a program. This tradition originates from the
seminal work by Manna and Waldinger~\cite{Manna-Waldinger:TOPLAS80},
and similar ideas has been used in more recent synthesis
work~\cite{Kneuss-al:OOPSLA13,Delaware-al:POPL15,Polikarpova-al:PLDI16}.
In particular, the overall structure of our synthesis algorithm
(backtracking and-or search) is similar to
\tname{Leon}~\cite{Kneuss-al:OOPSLA13}, but our rules focus on heap
manipulation, whereas their rules focus on synthesizing pure terms (so
in fact Leon can be used as a component by our algorithm). Recent work
on \tname{Optitian}~\cite{Miltner-al:POPL18} is very different in
scope---%
they synthesize bijective string lenses from regular expression
specifications and examples---%
but has interesting similarities in the technique. Their pre- and
post-condition are regexes, and their technique tries to ``align''
them by \eg, unfolding the Kleene star; this is similar to how \tool
tries to align the spatial pre- and post-condition by unfolding
predicates.

Deductive synthesis is closely related to proof search, and there has
been recent resurgence in applying proof-theoretic techniques, like
focusing, to program
synthesis~\cite{Frankle-al:POPL16,Scherer:SNAPL17}. But none of them
do it for a complex logic that can reason about stateful programs.

Despite the vast space of available tools for symbolic verification
based on Separation Logic:
\tname{Smallfoot}~\cite{Berdine-al:FMCO05},
\tname{HTT}~\cite{Nanevski-al:POPL10},
\tname{Bedrock}~\cite{Chlipala:PLDI11},
\tname{SLAyer}~\cite{Berdine-al:CAV11},
\tname{HIP/SLEEK}~\cite{Chin-al:OOPSLA11},
\tname{VeriFast}~\cite{Jacobs-al:NFM11},
\tname{Cyclist}~\cite{Rowe-Brotherston:CPP17},
\tname{SLAD}~\cite{Bouajjani-al:ATVA12},
\tname{GRASShopper}~\cite{Piskac-al:TACAS14},
\tname{Viper}~\cite{Mueller-al:VMCAI16},
to name just a few, to the best of our knowledge none of them has been
employed for deriving programs from specifications.
It is certainly our hope that this work will bring new synergies
between the research done in verification, theorem proving, and
program synthesis communities.


For instance, in our approach to establish termination of
\logic-synthesized programs, we used techniques close in spirit to the
methods for proving total correctness in type/SL-based frameworks.
\Eg, \logic's tags might be seen as a variant of resource capacities
used in \tname{HIP/SLEEK}~\cite{Le-al:ICFEM14}.
Our use of Definition~\ref{def:gvalidity} of sized validity is similar
to the induction on the finiteness of the heap used
by~\citet{Le-Hobor:ESOP18} in their work on a logic for fractional
shares.
Tagged predicates we use are reminiscent to the $\later$-modality in
type theories for state and recursion~\cite{Appel-al:POPL07}.

\section{Conclusion}
\label{sec:conclusion}

In their seminal paper,~\citet{Manna-Waldinger:TOPLAS80} set forth an
agenda for deductive synthesis of functional programs: \emph{``theorem
  provers have been exhibiting a steady increase in their
  effectiveness, and program synthesis is one of the most natural
  application of those systems''}.

In this work, we moved this endeavour to an uncharted territory of
stateful computations.
For this, we employed a proof system which, instead of a pure type
theory~\cite{Martin-Loef:84}, is based on Separation Logic---a
\emph{Type Theory of State}~\cite{Nanevski:OPLSS16}.
Taking this vision as a guiding principle, we designed \logicfull---a
modest extension of Separation Logic, tailored for program synthesis,
and implemented a proof search algorithm for it.
In doing so, we took full advantage of the power of \emph{local
  reasoning} about state~\cite{OHearn-al:CSL01}, which resulted in a
\emph{principled} and \emph{fast} approach for synthesyzing provably
correct heap-manipulating programs.

\begin{acks}
  We wish to thank Aquinas Hobor and Reuben Rowe for their insightful
  and extremely detailed comments on presentation, formalism, and
  examples. Their technical feedback has helped immensely to bring out
  the best of this paper.
  We also wish to thank Shachar Itzhaky and Ranjit Jhala for their
  comments on the draft.
  We are very grateful to Franck Cassez for his help with configuring
  and using \tname{ScalaSMT}.
  Sergey's research was supported by a generous gift from Google.

\end{acks}


\appendix







\section{Proofs of Formal Guarantees for the Synthesized Programs}
\label{sec:appendix}

\begin{theorem}[Theorem~\ref{thm:soundness}, Soundness of \logic]
  \label{thm:asoundness}
For any $n$, $\fctx'$, if
\begin{enumerate}[label=\emph{(\roman*)}]

\item\label{as1} $\ctx'; \env; \trans{\asn{\slP}}{\asn{\slQ}}{\prog}$ for a goal named
$f$ with formal parameters~$\env \eqdef \many{x_i}$, and 

\item\label{as2} $\ctx'$ is such that $\coh{\fctx', \ctx', n}$, and

\item\label{as3} for all $p^0(\many{e_i}), \phi; P$,
such that $\asn{\slP} = \asn{\phi; p^0(\many{e_i}) \sep P}$,
taking $\fun \eqdef f(\many{x_i}) : \asn{\phi; p^1(\many{e_i}) \sep
  \bump{P}{}}\asn{\bump{\slQ}{}}$,\\
$\ctx', \fun; \env; \asn{\slP}~\prog~\asn{\slQ}$ is $n'$-valid for
all $n' < n$ wrt.
$\fctx \eqdef \fctx', f~(\many{t_i~x_i})~\set{~\prog~}$,
\end{enumerate}
then $\ctx'; \env; \asn{\slP}~\prog~\asn{\slQ}$ is $n$-valid \wrt
$\fctx$.
\end{theorem}
\begin{proof}[Proof]

  Since $\prog$ has been produced by an \logic derivation, it is
  (\code{;})-associated to the right. 

  The proof is by the outer-level (well-founded) induction on the
  ``footprint size'' $n$. For each fixed $n$, the proof is by an
  inner-level induction on the structure of derivation
  $\ctx'; \env; \trans{\asn{\slP}}{\asn{\slQ}}{\prog}$, with symbolic
  one-step execution (via~$\Opsem$). In that induction step we rely on
  the fact that the residual ``suffix'' program is already valid, so
  we only need to show that applying a single rule preserves the
  validity, while ``appeding'' a command $\prog$ in front of an
  already synthesized program.
    
  The inner proof routine is straightforward for or (sub-)goals, where
  the spatial part of $\slP$ contains no predicate instances.

  All \emph{non-operational} rules (\framer, \pickr, \sleft, \empr,
  \nilvalr, \etc) either instantiate existentials in the post, weaken the pure
  obligations, or reduce the footprint, which is sound due to locality
  of reading and writing~\cite{OHearn-al:CSL01} and immutability of
  local variables; the soundness of \emph{operational} (\readr, \writer,
  \etc) rules is due to the soundness of SL. Specifically:

\begin{itemize}
  
\item \textbf{Case:}~\empr. The result follows immediately from
  Definition~\ref{def:gvalidity} of validity for any~$n$.
  
\item \textbf{Case:}~\sleft. The result follows immediately from
  Definition~\ref{def:gvalidity} of validity for any~$n$.
  
\item \textbf{Case:}~\nilvalr. The non-nullness is encoded by the
  spatial precondition part~\cite{Reynolds:LICS02}, so the
  transformation is sound. The result follows immediately from
  Definition~\ref{def:gvalidity} of validity for any heap size~$n$.
  
\item \textbf{Case:}~\starpar. The inequality of the two locations is
  encoded by the spatial part of the precondition, so the
  transformation is sound. The result follows immediately from
  Definition~\ref{def:gvalidity} of validity for any heap size~$n$.
  
\item \textbf{Case:}~\incons. The result follows immediately from
  Definition~\ref{def:gvalidity} of validity for any~$n$, as no
  physical pre-heap does satisfy $\phi$, so it is safe to emit
  \code{error}.
  
\item \textbf{Case:}~\framer. The \framer rule in \logic is
  different from its counterpart in vanilla SL in that it does not
  require a side condition saying that $R$ must not contain program
  variables that are modified by the program to be synthesized.
  Due to locality of read/write/malloc/free operations, they could
  not depend on $R$, as already synthesized in the residual program,
  hence adding $R$ to post will not break validity, given that $R$
  is also added to the precondition.
  
\item \textbf{Case:}~\readr. The synthesized read is safe, as the
  corresponding pointer is present in the precondition. Furthermore,
  it creates a program-level variable, which substantiates the
  environment variable $y$, used for the synthesis in the subgoal,
  hence the new program is valid according to
  Definition~\ref{def:gvalidity}.
  
\item \textbf{Case:}~\allocr. The subgoal has been synthesized in
  the larger footprint, afforded by allocation. The soundness of the
  rule follows from the axiomatization of
  \code{malloc}~\cite{OHearn-al:CSL01}.
  
\item \textbf{Case:}~\freer. Since the subgoal's result is valid in
  a smaller pre-heap footprint and does not access any variables
  from $R$, it is safe to deallocate it, hence the resulting program
  is valid.
  
\item \textbf{Case:}~\abductr. The validity is by the properties of
  the sequential composition $\prog_1; \prog_2$ and the soundness of
  the frame rule of SL, applied as follows to the first sub-goal:
  
  \[
    {{
        \frac{
          \ctx; \env; \trans{\asn{\phi; F}}{\asn{\phi; F'}}{\prog_1}  
        }
        {
          \ctx; \env; \trans{\asn{\phi; F \sep P \sep R}}{\asn{\phi;
              F' \sep P \sep R}}{\prog_1}  
        }
      }}
  \]
  \vspace{5pt}
  
\item \textbf{Case:}~\hunir. Since the unification only instantiates
  existentials, making the goal more specific, and unified variables
  will not appear in the program freely, the validity is by soundness
  of the SL's Auxiliary Variable Elimination
  rule~\cite{OHearn-al:CSL01}.
  
\item \textbf{Case:}~\pickr. By the soundness of substitution and
  auxiliary variable elimination.

\item \textbf{Case:}~\hyp. By the soundness of substitution and
  auxiliary variable elimination.

\item \textbf{Case:}~\sright. By the soundness of substitution and
  auxiliary variable elimination.
  
\item \textbf{Case:}~\closer. Validity follows from the fact that the
  sub-goal is valid in the presence of the elaborated postcondition
  $\asn{\psi \wedge [\subst]\sel_k \wedge [\subst]\chi_k; Q \sep R'}$.
  By the definition of $\satsl{}{}$, this heap satisfies the
  corresponding predicate definition in the initial goal's
  postcondition.
  
\end{itemize}

\noindent
In the case of predicate instances in the precondition, we take
advantage of the tag machinery and the
assumptions~\emph{\ref{as2}}--\emph{\ref{as3}} of the Theorem's
statement.

\begin{itemize}
  
\item \textbf{Case:}~\mkind. 
  The assumption~\emph{\ref{as2}} of the inductive hypothesis is
  satisfied via the initial assumption~\emph{\ref{as3}} and 1-tagging
  of the recursive hypothesis's precondition, so it will only be
  applicable to heaps of a size \emph{strictly smaller} than $n$
  (follows from \invkind and \applyr).
  
\item \textbf{Case:}~\invkind. The validity is due to the fact that
  the emitted \code{if}-\code{else} will combine the
  sub-programs, that are already valid \wrt their pres/posts. 
  The access to those elaborated pres/posts, ensuring the
  satisfiability of the elaborated precondition in each branch is due
  to the conditionals $\sel_j$, which are the same as the corresponding
  guard of the predicate instance being opened.

\item \textbf{Case:}~\applyr. This is the most interesting case. 
  Validity of a function body's substitution is via~\emph{\ref{as2}}.
  Notice that \applyr cannot be the first rule applied in the
  derivation, as then all predicate instances are 0-tagged, hence it
  should follow \invkind, which unfolds the instances, exposing
  smaller sub-heaps captured by their clauses (due to well-foundedness
  of the predicates), which are now amenable for using \applyr on
  them.
 
  The remaining reasoning differs depending on whether the function
  $f$ being called is a user-provided auxiliary one, or a recursive
  call to the top-level goal.
  If a function being applied is a user-provided
  specification (\ie, not a recursive self), then its validity is
  asserted by the second clause of Definition~\ref{def:coh}, otherwise
  the call is recursive, and will only take place on a smaller heap,
  satisfying the third option of Definition~\ref{def:coh}.
  
  \end{itemize}

\end{proof}


\end{document}